\documentclass[journal,10pt]{IEEEtran}

\usepackage{amsmath}
\usepackage{amsfonts}
\usepackage{mathrsfs}
\usepackage{amssymb}
\usepackage{bm}
\usepackage{cases}
\usepackage{stmaryrd}
\usepackage{mathrsfs}
\usepackage{bbm}
\usepackage[dvips]{graphicx}
\usepackage{CJK}
\usepackage{amsthm}
\usepackage{flushend}
\usepackage{algorithm}
\usepackage{algorithmic}
\usepackage{multirow}
\usepackage{array}
\usepackage{graphicx}
\usepackage{setspace}
\usepackage{subfigure}
\usepackage{epstopdf}
\usepackage{float}
\usepackage{cite}
\usepackage{enumerate}
\usepackage{color}

\newtheorem{pp}{Proposition}
\newtheorem{lemma}{\textbf{Lemma}}
\newtheorem{myDef}{\textbf{Definition}}
\newtheorem{remark}{\textbf{Remark}}

\begin{document}



\title{Secure Users Oriented Downlink MISO NOMA}
\author{
	Hui-Ming~Wang,~\IEEEmembership{Senior Member,~IEEE,}~%
	Xu~Zhang,~
	Qian~Yang,~\\
    and~Theodoros~A.~Tsiftsis,~\IEEEmembership{Senior Member,~IEEE}%
\thanks{Hui-Ming Wang, Xu Zhang, and Qian Yang are with the
School of Electronics and Information Engineering, and also with the Ministry
of Education Key Lab for Intelligent Networks and Network Security, Xi'an Jiaotong
University, Xi'an, 710049, Shaanxi, P. R. China. Email: {\tt
xjbswhm@gmail.com; jcx8008208820@stu.xjtu.edu.cn; yangq36@gmail.com}.}%
\thanks{Theodoros A. Tsiftsis is with the School of Electrical and Information Engineering, Jinan University (Zhuhai Campus), Zhuhai 519070, P. R. China. Email: {\tt theo\_tsiftsis@jnu.edu.cn}.}%
	}
\maketitle

\begin{abstract}

This paper proposes a secure users oriented multiple-input and single-output (MISO) non-orthogonal multiple access (NOMA) downlink transmission scheme, where multiple legitimate users are categorized as quality of service (QoS)-required users (QU) and the security-required users (SU) overheard by a passive eavesdropper. The basic idea is to exploit zero-forcing beamforming (ZFBF) to cancel interference among SUs, and then several QUs are efficiently scheduled based on the obtained beamforming vectors to divide the legitimate users into several user clusters, in such a way that the QUs could share the concurrent transmissions and help to interfere with the eavesdropper to enhance SU secrecy. The goal is to maximize the achievable minimum secrecy rate (MSR) and sum secrecy rate (SSR) of all SUs, respectively, subject to the secrecy outage probability (SOP) constraint of each SU and the QoS constraint of each QU.
To provide a comprehensive investigation
we consider two extreme cases that the eavesdropper has perfect multiuser detection ability (lower bound of secrecy) or does not have multiuser detection ability (upper bound of secrecy).  In the lower bound case, the Dinkelbach algorithm and the monotonic optimization (MO)-based outer polyblock approximation algorithm are proposed to solve the max-min secrecy rate (MMSR) and max-sum secrecy rate (MSSR) problems, respectively. As for the upper bound case, an alternative optimization (AO)-based algorithm is proposed to solve the two non-convex problems. Finally, the superiority of the proposed cases to the conventional orthogonal multiple access (OMA) one is verified by numerical results.
\end{abstract}

\begin{IEEEkeywords}
Convex optimization, non-orthogonal multiple access, power allocation, physical layer security.
\end{IEEEkeywords}


\section{Introduction}
Nowadays, the spectral efficiency becomes one of the key challenges to handle the increasing demand of data traffic. Moreover, due to the explosion of emerging services such as Internet of Things (IoT), 5G cellular communication systems need to support massive connectivity to meet the demand for low latency devices and diverse service types \cite{ReferenceSurvey1}. To that end, non-orthogonal multiple access (NOMA) has been proposed as a promising technique to exploit limited resources more efficiently in a non-orthogonal manner \cite{ReferenceNOMA}. An existing dominant NOMA scheme is to serve multiple users simultaneously via power domain multiplexing. In the power-domain NOMA scheme, the targeted information-bearing signals of multiple users are superimposed for transmission and efficient multiuser detection technology is exploited at the receiver via successive interference cancellation (SIC). Generally, the multiple NOMA users should be ordered according to their different channel conditions to facilitate the effective SIC process. In terms of the information theory, the sum capacity superiority of NOMA over OMA becomes more significant when the channel conditions of the concurrently served users become more distinguishable. Therefore, it is beneficial to schedule NOMA users with different channel conditions.  Moreover, users can also be distinguished according to their different QoS requirements \cite{ReferenceQoSordering1, ReferenceQoSordering2}.
The superiority of NOMA to conventional orthogonal multiple access (OMA) has been studied in the perspectives of spectral efficiency, energy efficiency, and fairness \cite{Reference1, Reference2, Reference3}.

Due to the broadcast nature of wireless communications, the confidential message is vulnerable to eavesdropping, which brings about the security challenge of wireless transmission. Physical layer security (PLS) has been regarded as a key complementary technology to safeguard communication security by exploiting the inherent random characteristics of  wireless physical-layer channels \cite{ReferencePLS1, ReferencePLS2, ReferencePLS3}. It is known that PLS mainly relies on the difference of the receive signal-to-noise ratios (SNRs) between the legitimate users and the eavesdropper. Therefore, introducing interference may be beneficial for secrecy enhancement. It is interesting to observe that the non-orthogonal resource allocation policy will introduce extra inter-user interference according to the NOMA principle, which may improve the PLS performance due to the significant degradation on the receive SNR of the eavesdropper. In this regard, employing the NOMA technology will not only increase the spectral efficiency, but improve the secrecy performance concurrently.

Due to the aforementioned advantages, the PLS topic in NOMA systems has drawn significant attention recently. The problem to maximize the achievable secrecy sum rate (SSR) in a downlink single-input single-output (SISO) NOMA system was firstly studied in \cite{Reference7}. Since then, several analyses and optimization problems of secrecy transmissions in secure downlink SISO NOMA systems have been studied in \cite{ReferenceSISO1, ReferenceSISO2, ReferenceSISO3, ReferenceSISO4, ReferenceSISO5}. Furthermore, using the multiple-antenna technology, the SSR optimization problems were studied in downlink multiple-input single-output (MISO) and multiple-input multiple-output (MIMO) NOMA systems by beamforming and power allocation in \cite{ReferenceMIMO1, ReferenceMIMO2, ReferenceMIMO3, ReferenceMIMO4}. However, the perfect instantaneous CSI of the eavesdroppers is required for the designs and optimizations in the above multi-antenna works, which is difficult to obtain in practice.

On the other hand, the assumption that the statistical CSI of eavesdroppers is available has been more widely adopted in current bibliography, and various NOMA scenarios and transmission schemes have been investigated from the perspective of PLS, e.g., mixed multicasting and unicasting \cite{ReferenceMIMO8, ReferenceMIMO9}, transmit antenna selection (TAS) \cite{ReferenceMIMO10, ReferenceMIMO11}, artificial noise (AN) \cite{ReferenceMIMO12}, and large-scale networks \cite{ReferenceMIMO7}. However, all the above works only focus on the secrecy performance analysis without optimizing the power allocation that by proper design can further improve the secrecy performance. Recently, Chen \emph{et al.} \cite{ReferenceMIMO15} investigated the ergodic secrecy rate of massive multi-antenna NOMA systems and optimized the power allocation for security enhancement. However, the authors assumed that the eavesdroppers did not have the capability of multi-user detection, which may underestimate their wiretapping abilities and lead to an over optimistic secrecy performance. Furthermore, the conclusions are based on the approximated lower bound of the ergodic SSR by exploiting the characteristics of a massive number of antennas, which is a special case that cannot apply to the general scenario.

Based on the above discussion, we investigate a novel downlink secure MISO-NOMA system with different categories of users. We consider that there are two categories of  legitimate users according to their different service requirements: 1) Users with secrecy requirements during the transmissions; 2) Users with only quality of service (QoS) requirements. Hereinafter these two kinds of users are called as security-required users (SU) and QoS-required users (QU), respectively. This scenario commonly exists in many practical applications. For example, some users with high secrecy priority (government officers, etc) will buy the additional PLS service from operators while some others may not. Another example is IoT applications such as the Internet of vehicles, where some sensors require confidential data (e.g., states of vehicle engines) while the data for some other sensors may have low or no secrecy requirements (e.g., temperature).

Under the above model, we propose to group the legitimate users into multiple clusters with one SU and multiple QUs in each cluster. The SU in each cluster has the highest priority so the secrecy requirement should be satisfied and the secrecy throughput should be optimized. Additionally, the concurrent transmission of the QU signals in a NOMA manner introduces interferences to enhance the SU's security under the condition that their QoS requirements are satisfied.

The main contributions of our paper can be summarized as follows:
\begin{enumerate}[1)]
\item Subject to the secrecy outage probability (SOP) constraints of SUs and QoS requirements of QUs, we investigate the beamforming and power allocation schemes for the max-min secrecy rate (MMSR) and max-sum secrecy rate (MSSR) problems, respectively. Effective user scheduling schemes are proposed as well. Furthermore, according to different wiretap capabilities of the eavesdropper, we discuss both cases that the eavesdropper applies perfect SIC and no SIC, which correspond to the lower and upper bounds of the secrecy performance, respectively.

\item For the MMSR problem with a focus on the fairness among all the SUs with different channel conditions, an efficient Dinkelbach algorithm is proposed to obtain the globally optimal solution for the lower bound case, while an alternative optimization (AO)-based algorithm is developed to derive the sub-optimal solution for the upper bound case, respectively.

\item For the MSSR problem aiming to enhance the confidential transmission rate, an efficient outer polyblock approximation based monotonic algorithm is exploited to obtain the globally optimal solution for the lower bound case, while the AO-based algorithm is exploited to derive the sub-optimal solution for the upper bound case, respectively.

\end{enumerate}



The rest of this paper is organized as follows: In Section \uppercase\expandafter{\romannumeral2}, we present the system model and propose the transmission and user scheduling schemes for the considered system. We focus on the MMSR and MSSR problems in Section \uppercase\expandafter{\romannumeral3} and \uppercase\expandafter{\romannumeral4}, respectively. Numerical results are presented in Section \uppercase\expandafter{\romannumeral5} before the conclusion is drawn in Section \uppercase\expandafter{\romannumeral6}.

\emph{Notations:} Boldface lower-case and upper-case letters denote vectors and matrices, respectively. The transpose, conjugate, conjugate transpose, pseudo-inverse, Euclidean norm, Frobenius norm, trace, and the largest eigenvalue of the matrix ${\bf{A}}$ are denoted as ${{\bf{A}}^T}$, ${{\bf{A}}^*}$, ${{\bf{A}}^H}$, ${{\bf{A}}^\dag }$, $\left\| {\bf{A}} \right\|$, ${\left\| {\bf{A}} \right\|_F}$, $Tr\left( {\bf{A}} \right)$, and ${\lambda _{\max }}\left( {\bf{A}} \right)$, respectively. ${\left\{ A \right\}^ + }$ denotes the operation $\max \left\{ {A,0} \right\}$. $CN\left( {0,{\sigma _N^2}} \right)$ denotes zero-mean additive white Gaussian noise (AWGN) with variance $\sigma _N^2$

\section{System Model and the Proposed Scheme}
In this section, we describe the system model, the transmission and user scheduling schemes and the secrecy performance metric.
%
\subsection{System Model}
Consider a downlink MISO-NOMA system as shown in Fig. \ref{f_System_model}. Since 5G cellular communication systems need to support massive connectivity due to the explosion of emerging applications such as IoT, the base station (BS) is equipped with $N$ antennas and we make a practical assumption that there exists $M > N$ legitimate users (LUs) equipped with a single antenna. According to distinct service requirements of different kinds of users, the LUs are categorized as SUs with secrecy requirements and QUs with only QoS requirements. There is a passive eavesdropper (Eve) equipped with a single antenna which potentially overhears the targeted messages for all the SUs. We make a practical assumption that the instantaneous CSI of Eve is absent and only its statistical CSI is available. On the other hand, we assume that Eve knows the CSI of itself.

Based on the described user categorization, our goal is to achieve better secrecy performance of the SUs and to simultaneously satisfy the QoS requirements of the QUs, and, thus, the spectral efficiency of the system can be improved. To achieve this goal, the basic idea is to group the LUs into multiple clusters with one SU pairing with several QUs in each cluster. The detailed transmission and user scheduling schemes are discussed in Section \uppercase\expandafter{\romannumeral2}-B and Section \uppercase\expandafter{\romannumeral2}-C, respectively.

\begin{figure}[t]
\centering
\includegraphics[width=8cm,height=5cm]{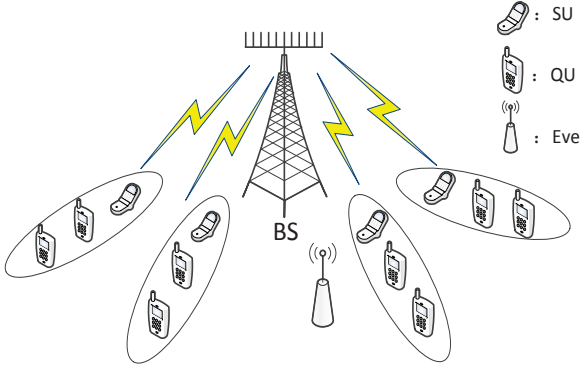}
\caption{The considered downlink cellular network with a multi-antenna BS, multiple single-antenna legitimate users categorized as SUs and QUs, and a passive eavesdropper Eve. The NOMA technology is exploited to simultaneously serve multiple users under each spatial beam.}\label{f_System_model}
\end{figure}

\subsection{Transmission Scheme}
Under the proposed cluster-based downlink NOMA system model, there is one SU and multiple QUs in each cluster. The BS fully exploits the spatial degree of freedom (DoF) to design the distinct beamforming vector for each user cluster, which means that the scheduled LUs in each cluster are served by a common beam in a NOMA manner \cite{ReferenceMIMO15, ReferenceSDMA}. To fully exploit the spatial DoF, we make an additional assumption that the number of LUs is $M = N \times K$ and $M$ LUs are divided into $G = N$ clusters with $K$ users in each cluster.
Since available resources are shared in the NOMA scheme, there may exist both intra-cluster interference and inter-cluster interference between LUs. Since the SU in each cluster has the highest priority to achieve secrecy transmission, we propose a zero-forcing beamforming (ZFBF) scheme with respect to all the SUs \cite{ReferenceZFBF}, i.e., multiple SUs in different clusters have no inter-cluster interference. This scheme has the following two merits:
\begin{enumerate}[1)]
\item The ZFBF scheme alleviates the effect of inter-cluster interference of the SUs, which is beneficial to enhance the signal-to-interference-and-noise ratios (SINRs) and, hence, the secrecy performance can be improved.
\item In each cluster, the SU and the QUs share the available resources in a non-orthogonal way. Except for the spectral efficiency enhancement, the intra-cluster interference in each cluster will confuse the potential Eve in order to further improve the secrecy throughput.
\end{enumerate}

All the wireless channels are assumed to experience quasi-static Rayleigh block fading. Specifically, the independent $N \times 1$ channel coefficient vector from the BS to the $k$-th LU in the $g$-th cluster is denoted as ${{\bf{h}}_{g,k}}\sim CN\left( {0,{{\bf{I}}_N}} \right)$. Without loss of generality, the SU is represented by the first user in each cluster. Define
\begin{equation}\label{e1}
  {{\bf{H}}_1} = {\left[ {{{\bf{h}}_{1,1}},{{\bf{h}}_{2,1}},...,{{\bf{h}}_{G,1}}} \right]^T},
\end{equation}
and then by employing the ZFBF scheme, the beamforming matrix is given as
\begin{equation}\label{e2}
  {\bf{\tilde W}} = \left[ {{{{\bf{\tilde w}}}_1},{{{\bf{\tilde w}}}_2},...,{{{\bf{\tilde w}}}_G}} \right] = {{\bf{H}}_1}^*{\left( {{{\bf{H}}_1}{{\bf{H}}_1}^*} \right)^{ - 1}}.
\end{equation}
We define ${{\bf{w}}_g} = {{{{{\bf{\tilde w}}}_g}} \mathord{\left/
 {\vphantom {{{{{\bf{\tilde w}}}_g}} {\left\| {{{{\bf{\tilde w}}}_g}} \right\|}}} \right.
 \kern-\nulldelimiterspace} {\left\| {{{{\bf{\tilde w}}}_g}} \right\|}}$ as the $N \times 1$ normalized beamforming vector for the $g$-th cluster. By applying the NOMA protocol, the superimposed downlink signal transmitted by the BS is formulated as
\begin{equation}\label{e3}
  {\bf{x}} = \sum\limits_{g = 1}^G {{{\bf{w}}_g}{x_g}}  = \sum\limits_{g = 1}^G {{{\bf{w}}_g}\left( {\sum\limits_{k = 1}^K {\sqrt {{\alpha _{g,k}}P} {s_{g,k}}} } \right)},
\end{equation}
where ${x_g}$ is the transmitted superposition signal for the $g$-th cluster, $P$ is the total power consumption at the BS, ${\alpha _{g,k}}$ is the power allocation coefficient for the $k$-th user in the $g$-th cluster, and ${s_{g,k}}$ is the targeted signal for the $k$-th user in the $g$-th cluster which satisfies $E\left\{ {{{\left| {{s_{g,k}}} \right|}^2}} \right\} = 1$.

\subsection{User Scheduling-Based SIC}
As described in the above two subsections, $G$ clusters are grouped, and one SU and $K-1$ QUs in each cluster are served by a common beam. Under this condition, since a random division may lead to a poor performance of our proposed scheme, the effective user scheduling should be carefully designed to achieve better secrecy performance and enhance the spectral efficiency. Specifically, $G$ SUs are firstly scheduled and the corresponding beamforming vectors are designed in (\ref{e2}) according to their instantaneous CSI. Then based on the obtained beamforming vectors ${{\bf{w}}_g}$, $K-1$ QUs with similar channel fading vectors to the SU are scheduled in each cluster to alleviate the effect of inter-cluster interference. In addition, since the SU in each cluster has the highest priority, the QUs are scheduled to simultaneously ensure that the SU has the channel superiority in each cluster for secrecy performance enhancement, namely,
\begin{equation}\label{e5}
  {\left| {{\bf{h}}_{g,1}^H{{\bf{w}}_g}} \right|^2} > {\left| {{\bf{h}}_{g,2}^H{{\bf{w}}_g}} \right|^2} >  \cdot  \cdot  \cdot  > {\left| {{\bf{h}}_{g,K}^H{{\bf{w}}_g}} \right|^2}.
\end{equation}
where the ordering of the effective channel gains of the QUs in each cluster is assumed without loss of generality. Based on the above user scheduling policy, the receive signal of the $k$-th user in the $g$-th cluster is expressed as
\begin{equation}\label{e4}
  {y_{g,k}} = {\bf{h}}_{g,k}^H\sum\limits_{i = 1}^G {{{\bf{w}}_i}\left( {\sum\limits_{j = 1}^K {\sqrt {{\alpha _{i,j}}P} {s_{i,j}}} } \right)}  + {n_{g,k}},
\end{equation}
where ${n_{g,k}} \sim CN\left( {0,{\sigma _N^2}} \right)$ is the additive noise.

According to the NOMA principle, strong users in each cluster have the ability to avoid the interference caused by the weaker users and the SIC method is exploited. Specifically, the $n$-th user will firstly detect and eliminate all the targeted signals of the $m$-th user ($m > n$) in a successive way, and then the targeted signal of the $n$-th user is detected by treating the remaining other users' signals as noise. In our case, the SINR of each LU in the $g$-th cluster is
\begin{equation}\label{e6}
  {\rm{SIN}}{{\rm{R}}_{g,1}} = {\left| {{\bf{h}}_{g,1}^H{{\bf{w}}_g}} \right|^2}\rho {\alpha _{g,1}},
\end{equation}
\begin{equation}\label{e7}
  {\rm{SIN}}{{\rm{R}}_{g,k}} = \frac{{{S_{g,k}}}}{{I_1^{g,k} + I_2^{g,k} + 1}},\;k = 2,...,K,
\end{equation}
where $\rho  = \frac{P}{{{\sigma _N^2}}}$ denotes the average transmit SNR at the BS, ${S_{g,k}} = {\left| {{\bf{h}}_{g,k}^H{{\bf{w}}_g}} \right|^2}\rho {\alpha _{g,k}}$ is the receive power of the desired signal, $I_1^{g,k} = {\left| {{\bf{h}}_{g,k}^H{{\bf{w}}_g}} \right|^2}\sum\nolimits_{j = 1}^{k - 1} {\rho {\alpha _{g,j}}} $ accounts for the partially cancelled intra-cluster interference after SIC, $I_2^{g,k} = \sum\nolimits_{i = 1,i \ne g}^G {{{\left| {{\bf{h}}_{g,k}^H{{\bf{w}}_i}} \right|}^2}\sum\nolimits_{j = 1}^K {\rho {\alpha _{i,j}}} } $ denotes the inter-cluster interference, and the third term is the power of normalized AWGN noise.

\begin{remark}
  \textnormal{Though the optimal decoding order or user pairing for NOMA has been investigated in \cite{ReferenceSISO2, ReferenceOptimal1, ReferenceOptimal2}, these works are generally focused on the single-antenna scenarios. In contrast, the optimal design in our considered multi-antenna user cluster based system is quite complicated. Due to the high complexity of the optimal design, we have discussed a specific and novel secure users oriented transmission and user scheduling scheme and the simulation results in Section \uppercase\expandafter{\romannumeral5} demonstrate that under some certain conditions, the secrecy performance of our proposed scheme is better than that of traditional beamforming design only. The optimal design of the problem is beyond the scope of this paper and will be investigated in our future work.}
\end{remark}

\subsection{Eavesdropping Model}
If the targeted user is the SU in the $g$-th cluster, the receive signal at Eve is formulated as
\begin{equation}\label{e8}
  {y_e^g} = s_{e,1}^g + s_{e,2}^g + s_{e,3}^g + {n_e},
\end{equation}
where $s_{e,1}^g = {\bf{h}}_e^H{{\bf{w}}_g}\sqrt {{\alpha _{g,1}}P} {s_{g,1}}$ denotes the desired signal of Eve, $s_{e,2}^g = {\bf{h}}_e^H{{\bf{w}}_g}\sum\nolimits_{k \ne 1}^K {\sqrt {{\alpha _{g,k}}P} {s_{g,k}}} $ and $s_{e,3}^g = {\bf{h}}_e^H\sum\nolimits_{i = 1,i \ne g}^G {{{\bf{w}}_i}\left( {\sum\nolimits_{k = 1}^K {\sqrt {{\alpha _{i,k}}P} {s_{i,k}}} } \right)} $ are the interference from the signals of the intra-cluster users and inter-cluster users, respectively, and ${n_{e}} \sim CN\left( {0,{\sigma _N^2}} \right)$ is the additive noise. ${{\bf{h}}_e} \sim CN\left( {0,{{\bf{I}}_N}} \right)$ is the $N \times 1$ channel coefficient vector from the BS to Eve, and the statistical CSI of Eve is available at the BS according to Section \uppercase\expandafter{\romannumeral2}-A.

From (\ref{e8}), both intra-cluster and inter-cluster interferences exist when Eve detects the targeted signals of SUs. In order to provide a comprehensive investigation on the impact of such interferences on the secrecy performance, two extreme cases are considered according to different assumptions about the multiuser detection ability of Eve.
\begin{enumerate}[1)]
\item \textbf{Lower Bound of secrecy performance:} In this case, we consider the worst-case scenario that Eve has powerful multiuser detection ability, which is a commonly used assumption on the detection ability of Eve in the existing NOMA PLS literatures \cite{Reference7, ReferenceSISO1, ReferenceSISO2, ReferenceSISO3}. Specifically, Eve can distinguish the multiuser data streams from BS and thus decode each target signal only under the impact of noise without interference, which leads to the theoretical lower bound of the secrecy performance with respect to the detection abilities of Eve. The SNR of detecting the targeted signal of the SU in the $g$-th cluster is formulated as
\begin{equation}\label{e9}
  {\rm{SINR}}_g^{\left( e \right)} = {\left| {{\bf{h}}_e^H{{\bf{w}}_g}} \right|^2}\rho {\alpha _{g,1}}.
\end{equation}
\item \textbf{Upper Bound of secrecy performance:} In this case, we assume that Eve has no multiuser detection ability. Under this condition, Eve detects the targeted signal under both the AWGN noise and all the interferences. The SINR of detecting the targeted signal of the SU in the $g$-th cluster is formulated as
\begin{equation}\label{en9}
  {\rm{SINR}}_g^{\left( e \right)} = \frac{{S_g^{\left( e \right)}}}{{{\rm{I}}_{g,1}^{\left( e \right)} + {\rm{I}}_{g,2}^{\left( e \right)} + 1}},
\end{equation}
where $S_g^{\left( e \right)} = {\left| {{\bf{h}}_e^H{{\bf{w}}_g}} \right|^2}\rho {\alpha _{g,1}}$ is the receive power of the targeted signal, ${\rm{I}}_{g,1}^{\left( e \right)} = {\left| {{\bf{h}}_e^H{{\bf{w}}_g}} \right|^2}\sum\nolimits_{k \ne 1}^K {\rho {\alpha _{g,k}}} $ denotes the intra-cluster interference and ${\rm{I}}_{g,2}^{\left( e \right)} = \sum\nolimits_{i = 1,i \ne g}^G {{{\left| {{\bf{h}}_e^H{{\bf{w}}_i}} \right|}^2}\sum\nolimits_{k = 1}^K {\rho {\alpha _{i,k}}} } $ accounts for the inter-cluster interference.
\end{enumerate}

\begin{remark}
  \textnormal{Since Eve is a passive eavesdropper overhearing the confidential transmissions, it is general and practical that we don't make a specific assumption on the location of Eve. In addition, it should be pointed out that our aforementioned secure users oriented transmission and user scheduling scheme is essentially channel-dependent rather than location-dependent, which also indicates that it does not matter whether Eve is located in the cluster or not.}
\end{remark}

\begin{remark}
  \textnormal{In practice, the detection capability of Eve always lies between the perfect cancellation and no cancellation of interferences. Therefore, the cases in which the intra-cluster interference and the inter-cluster interference can be partially subtracted and Eve detects the targeted signal under the impact of noise and residual interference are more practical. We note that it could be considered as a special case of the upper bound case, since the receive SINR of each SU takes the same form as (\ref{en9}) with only different terms of interference in the denominator. We will see later that our proposed optimization algorithms are also applicable to these cases.}
\end{remark}

\begin{remark}
  \textnormal{On the other hand, the SNR in (\ref{e9}) could be mathematically viewed as a special case of the SINR in (\ref{en9}) when ${\rm{I}}_{g,1}^{\left( e \right)}=0$ and ${\rm{I}}_{g,2}^{\left( e \right)}=0$, which implies that the optimization algorithm studied for the upper bound case also applies to the lower bound case. However, we will see in the following that with (\ref{e9}), the globally optimal solution could be achieved while this is not the case when dealing with (\ref{en9}). That is why we will discuss the lower bound case exclusively.}
\end{remark}

\begin{remark}
  \textnormal{Under some practical scenarios that Eve has partial detection capability such as SIC ability, the corresponding SINR has the same form as (10), where ${\rm{I}}_{g,1}^{\left( e \right)}$ and ${\rm{I}}_{g,2}^{\left( e \right)}$ represent the residual intra-cluster and inter-cluster interference after partial interference cancellation. As a result, all the cases that Eve has different partial detection capabilities can be considered as special ones of the upper bound case and our optimization algorithms are also applicable to these cases.}
\end{remark}

\subsection{Secrecy Performance Metrics}
If the instantaneous CSI of the wiretap channel is available at BS, the achievable secrecy rate ${R_{s,g}}$ of the SU in the $g$-th cluster can be defined as
\begin{equation}\label{e10}
  {R_{s,g}} \le {\left\{ {C\left( {{\rm{SIN}}{{\rm{R}}_{g,1}}} \right) - C\left( {{\rm{SINR}}_g^{\left( e \right)}} \right)} \right\}^ + },
\end{equation}
where $C\left( \rho  \right) \buildrel \Delta \over = {\log _2}\left( {1 + \rho } \right)$. However, since the instantaneous CSI of ${{\bf{h}}_e}$ is not available and is treated as a random variable as aforementioned in Section \uppercase\expandafter{\romannumeral2}-A, using (\ref{e10}) cannot directly determine the exact secrecy rate. Under this condition, given a constant $R_{s,g}$ (in the optimization problem $R_{s,g}$ is treated as a optimization variable) the SOP is introduced as
\begin{equation}\label{e11}
  {P_g} = \Pr \left\{ {C\left( {{\rm{SINR}}_g^{\left( e \right)}} \right) > C\left( {{\rm{SIN}}{{\rm{R}}_{g,1}}} \right) - {R_{s,g}}} \right\}.
\end{equation}
Based on the theory of the wiretap code \cite{ReferenceSISO2},  the positive difference of the maximum codeword transmission rate $C\left( {{\rm{SIN}}{{\rm{R}}_{g,1}}} \right)$ and the secrecy rate ${R_{s,g}}$ is the redundant rate cost to provide security against Eve, and SOP is defined as the probability that the Eve's channel capacity exceeds the
redundant rate as in (\ref{e11}).

As mentioned before, our goal is to maximize the secrecy rate $R_{s,g}$ under the SOP constraints of the SUs and the QoS constraints of the QUs. In the NOMA system, we consider two kinds of performance metrics, namely, the minimum secrecy rate (MSR) and SSR of all the SUs. Since the beamforming vectors have already been determined according to the zero-forcing criteria, our aim is to optimize the power allocation strategy. The specific problems will be introduced and efficiently solved in Sections \uppercase\expandafter{\romannumeral3} and \uppercase\expandafter{\romannumeral4}, respectively.

\section{Max-Min Achievable Secrecy Rate}
In this section, we design the optimal power allocation strategies in both the lower and upper bound cases respectively. In particular, we maximize the minimum secrecy rate of the SUs subject to the SOP constraints of SUs, the QoS requirements of QUs, and the total available transmit power constraint. The objective of MMSR reflects the fairness of the confidential transmission considering different CSI conditions of SUs.

\subsection{Problem Formulation}
According to the aforementioned requirements, the MMSR problem is formulated as follows:
\begin{subequations}\label{P1}
\begin{align}
&\mathop {\max }\limits_{{\bf{a}},{{R_{s,g}}}} \mathop {\min }\limits_{1 \le g \le G} {\left\{ {{R_{s,g}}} \right\}^ + }\label{P1e1}\\
&\;\;\;\textrm{s.t.}\;\;{P_g} \le \varepsilon ,\;g = 1,2,...G\label{P1e2} \\
&\;\;\;\;\;\;\;\;\;{\rm{SIN}}{{\rm{R}}_{g,k}}\! \ge\! {r_{g,k}},\;g\! = \!1,2,...,G;\;k = 2,...,K\label{P1e3}\\
&\;\;\;\;\;\;\;\;\sum\limits_{g = 1}^G {\sum\limits_{k = 1}^K {{\alpha _{g,k}}} }  \le 1,\label{P1e4}
\end{align}
\end{subequations}
where ${\bf{a}} = \left\{ {{\alpha _{g,k}},g = 1,2,...,G,k = 1,2,...,K} \right\}$ represents the set of all the power allocation coefficients. Constraint (\ref{P1e2}) includes the SOP constraints of all the SUs, where $\varepsilon $ is the maximum allowable SOP threshold. Constraint (\ref{P1e3}) includes the QoS requirements of all the QUs where $r_{g,k}$ is the SINR threshold, and (\ref{P1e4}) is the total power consumption constraint.

According to (\ref{e11}), ${{R_{s,g}}}$ should be lower than ${C\left( {{\rm{SIN}}{{\rm{R}}_{g,1}}} \right)}$ in order to satisfy the required SOP constraint, and thereby can be expressed as the form of ${C\left( {{\rm{SIN}}{{\rm{R}}_{g,1}}} \right)}$ subtracting a positive redundancy term. We introduce new slack variables ${\bf{z}} = \left\{ {{z_g} \ge 0,\;g = 1,2,...,G} \right\}$ and directly transform the SOP as
\begin{equation}\label{e12}
  {P_g} = \Pr \left\{ {{\rm{SINR}}_g^{\left( e \right)} > {z_g}} \right\},
\end{equation}
where ${z_g}$ satisfies the following inequality
\begin{equation}\label{e14}
  {R_{s,g}} \ge {\log _2}\left( {1 + {\rm{SIN}}{{\rm{R}}_{g,1}}} \right) - {\log _2}\left( {1 + {z_g}} \right).
\end{equation}
After transforming problem (\ref{P1}) into its epigraph form \cite{ReferenceConvex} and substituting (\ref{e6}), (\ref{e7}), (\ref{e12}) and (\ref{e14}) into the transformed problem, the MMSR problem can be reformulated as
\begin{subequations}\label{P2}
\begin{align}
&\mathop {\max }\limits_{{\bf{a}},{\bf{z}}} \mathop {\min }\limits_{1 \le g \le G} \left\{ {{{\log }_2}\left( {\frac{{1 + {{\left| {{\bf{h}}_{g,1}^H{{\bf{w}}_g}} \right|}^2}\rho {\alpha _{g,1}}}}{{1 + {z_g}}}} \right)} \right\}\label{P2e1}\\
&\;\;\;\textrm{s.t.}\;\Pr \left\{ {{\rm{SINR}}_g^{\left( e \right)} > {z_g}} \right\} \le \varepsilon ,\;g = 1,2,...,G\label{P2e2}\\
&\;\;\;\;\;\;\;\;\;\frac{{{S_{g,k}}}}{{I_1^{g,k} + I_2^{g,k} + 1}} \ge {r_{g,k}},\;\;k = 2,...,K,\; \textrm{and} \ (\textrm{\ref{P1e4}}).\label{P2e3}
\end{align}
\end{subequations}
We note that problem (\ref{P2}) is equivalent to problem (\ref{P1}). On one hand, since the left-hand side of constraint (\ref{P2e2}) is a monotonically decreasing function of ${{z_g}}$, this constraint is active at the optimum to reduce the redundancy term ${\log _2}\left( {1 + {z_g}} \right)$. Under this condition, by comparing the expressions of SOP in (\ref{e11}) and (\ref{e12}), we can obtain that the inequality (\ref{e14}) is active and thus the equivalence is proved. On the other hand, the ${\left\{  \cdot  \right\}^ + }$ operation can be ignored due to the fact that each term in the objective function must be non-negative at optimality. This is because if ${R_{s,g}} < 0$ for any $g$, we can always stop transmitting the desired signal for the SU in the $g$-th cluster.

According to Section \uppercase\expandafter{\romannumeral2}-C, since ${S_{g,k}} = {\left| {{\bf{h}}_{g,k}^H{{\bf{w}}_g}} \right|^2}\rho {\alpha _{g,k}}$, $I_1^{g,k} = {\left| {{\bf{h}}_{g,k}^H{{\bf{w}}_g}} \right|^2}\sum\nolimits_{j = 1}^{k - 1} {\rho {\alpha _{g,j}}} $ and $I_2^{g,k} = \sum\nolimits_{i = 1,i \ne g}^G {{{\left| {{\bf{h}}_{g,k}^H{{\bf{w}}_i}} \right|}^2}\sum\nolimits_{j = 1}^K {\rho {\alpha _{i,j}}} } $ are linear functions of power allocation coefficients, constraint (\ref{P2e3}) can be transformed as ${S_{g,k}} \ge {r_{g,k}}\left( {I_1^{g,k} + I_2^{g,k} + 1} \right)$, which is an inequality constraint composed of linear functions. Additionally, constraint (\ref{P1e4}) is also a linear inequality constraint with the sum of all the power allocation coefficients in the left-hand side. However, due to the linear-fractional form as well as the $\min \left(  \cdot  \right)$ operation in the objective function, and the probabilistic constraint in (\ref{P2e2}), problem (\ref{P2}) remains difficult to be directly solved.

\subsection{Proposed Method for the Lower Bound Case}
In the lower bound case, we first substitute (\ref{e9}) into (\ref{P2e2}), which is in a quadratic form, and we then approximate the probability constraint by a deterministic form to deal with its non-convexity. It can be easily observed that (\ref{P2e2}) can be rewritten as
\begin{equation}\label{e15}
  \Pr \left\{ {{\bf{\tilde h}}_e^H{\bf{\Sigma }}_g^{\left( l \right)}{{{\bf{\tilde h}}}_e} > {z_g}} \right\} \le \varepsilon ,
\end{equation}
where ${{\bf{\tilde h}}_e} \sim CN\left( {0,{{\bf{I}}_N}} \right)$ and ${{\bf{\Sigma }}_g^{\left( l \right)}} = {{\bf{w}}_g}{\bf{w}}_g^H\rho {\alpha _{g,1}}$. To further transform the SOP constraint, we introduce the following lemma \cite{Reference16}:
\begin{lemma}\label{BS_Lemma}
  (Bernstein-type inequality)  \textnormal{Let ${\bf{G}} = {{\bf{h}}^H}{\bf{Qh}}$ where ${\bf{h}} \sim CN\left( {0,{{\bf{I}}_N}} \right)$ and ${\bf{Q}} \in {{\rm{C}}^{N \times N}}$ is a Hermite matrix. Then for any $\sigma  > 0$, we have
  \begin{equation}\label{e16}
    \Pr \!\!\left\{ {{\bf{G}} \!\ge\! Tr\left( {\bf{Q}} \right)\!\! +\!\! \sqrt {2\sigma } {{\left\| {\bf{Q}} \right\|}_F}\!\! +\!\! \sigma \!\! \cdot\!\! {{\left\{ {{\lambda _{\max }}\left( {\bf{Q}} \right)} \right\}}^ + }} \right\}\!\! \le\!\! \exp \left( { - \sigma } \right).
  \end{equation}
  }
\end{lemma}
\begin{proof}
   The detailed proof can be seen in \textbf{Appendix A}.
\end{proof}

The Bernstein-type inequality is known as a standard technique to deal with probabilistic constraints involving the quadratic form of a Gaussian random vector \cite{Reference6}. Applying \emph{Lemma \ref{BS_Lemma}} and defining ${{\bf{\Theta }}_g} \buildrel \Delta \over = {\bf{\tilde h}}_e^H{\bf{\Sigma }}_g^{\left( l \right)}{{{\bf{\tilde h}}}_e}$, it is straight-forward to obtain that
\begin{equation}\label{e17}
  \Pr \!\!\left\{\!\! {{{\bf{\Theta }}_g} \!\ge\! Tr\!\left(\! {{\bf{\Sigma }}_g^{\left( l \right)}} \!\right)\!\! + \!\!\sqrt {2\sigma } {{\left\| \! {{\bf{\Sigma }}_g^{\left( l \right)}} \!\right\|}_F} \!\!+\!\! \sigma \!\! \cdot\!\! {{\left\{\!\! {{\lambda _{\max }}\left( {{\bf{\Sigma }}_g^{\left( l \right)}} \right)} \!\!\right\}}^ + }} \!\!\right\}\!\! \le\!\! \exp \!\left(\! { - \sigma }\! \right)\!.
\end{equation}
Then, by letting $\sigma = \ln \left( {\varepsilon ^{ - 1}} \right)$, we can approximately transform the SOP constraint into \footnote{Some detailed discussions will be presented in Remark 6.}
\begin{equation}\label{e18}
  {z_g} \ge Tr\left( {{\bf{\Sigma }}_g^{\left( l \right)}} \right) + \sqrt {2\sigma } {\left\| {{\bf{\Sigma }}_g^{\left( l \right)}} \right\|_F} + \sigma  \cdot {\left\{ {{\lambda _{\max }}\left( {{\bf{\Sigma }}_g^{\left( l \right)}} \right)} \right\}^ + }.
\end{equation}
Since ${{\bf{\Sigma }}_g^{\left( l \right)}}$ is a linear function of ${\alpha _{g,1}}$, the transformed SOP in (\ref{e18}) is convex.

We now deal with the non-convexity of the objective function. By exploiting the monotonicity of logarithmic function, it is observed that the objective function can be transformed as
\begin{equation}\label{e19}
  \mathop {\max }\limits_{{\bf{a}},{\bf{z}}} \mathop {\min }\limits_{1 \le g \le G} \left\{ {\frac{{1 + {{\left| {{\bf{h}}_{g,1}^H{{\bf{w}}_g}} \right|}^2}\rho {\alpha _{g,1}}}}{{1 + {z_g}}}} \right\}.
\end{equation}
Since the numerator and the denominator of (\ref{e19}) both are linear functions of the optimization variables, the transformed problem is a standard fractional programming problem which can be solved by the Dinkelbach algorithm \cite{Reference14} for obtaining the globally optimal solution. Specifically, we introduce the initialized ${\lambda } = 0$ as the iteration parameter. In each iteration, by utilizing the introduced $\lambda$, the optimization problem is transformed into
\begin{subequations}\label{P3}
\begin{align}
&\mathop {\max }\limits_{{\bf{a}},{\bf{z}},\tau } \;\tau \tag{22}\\
&\;\;\textrm{s.t.}\;\;\;1\!\! +\!\! {\left| {{\bf{h}}_{g,1}^H{{\bf{w}}_g}} \right|^2}\rho {\alpha _{g,1}} \!\!-\!\! {\lambda}\left( {1 + {z_g}} \right) \!\ge\! \tau ,g = 1,2,...,G, \label{P3e2}\notag\\
&\;\;\;\;\;\;\;\;\;(\textrm{\ref{P2e3}})~\textrm{and}~(\textrm{\ref{e18}}),\notag
\end{align}
\end{subequations}
which is a convex problem that can be solved by some efficient solvers such as CVX \cite{Reference17}. After finite steps of iterations, the obtained result converges to the globally optimal solution within the predefined error tolerance. The algorithm is summarized in \emph{Algorithm \ref{A1}}.
\begin{algorithm}
\caption{Dinkelbach Algorithm for the MMSR Problem in the Lower Bound Case}\label{A1}
\begin{algorithmic}[1]
\STATE ${\lambda _1} = 0$;
\STATE Set the allowable tolerance $\delta  \ll 1$. Initialize the iteration index $n = 1$ and ${\rm{Judgemark}} = 0$.
\STATE \textbf{Repeat}:
\STATE Solve the problem (\ref{P3}) with $\lambda  = {\lambda _n}$. Obtain the optimal solution $\alpha _{g,k}^{ * \left( n \right)}$ and $z_g^{*\left( n \right)}$ in the $n$-th iteration
\STATE \textbf{(Compare)} If
\begin{align}
  &\mathop {\min \!\!}\limits_{1 \le g \le G} \;\left\{ {1\!\! +\!\! {{\left| {{\bf{h}}_{g,1}^H{{\bf{w}}_g}} \right|}^2}\rho \alpha _{g,1}^{ * \left( {n{\rm{ - }}1} \right)}\!\! - \!\!{\lambda _n}\left( {1 \!\!+\!\! z_g^{*\left( {n{\rm{ - }}1} \right)}} \right)} \right\}\!\! \le\!\! \delta, \nonumber
\end{align}
set ${\rm{Judgemark}} = 1$.
\STATE Update ${\lambda _{n + 1}} = \mathop {\min }\limits_{1 \le g \le G} \;\frac{{1 + {{\left| {{\bf{h}}_{g,1}^H{{\bf{w}}_g}} \right|}^2}\rho \alpha _{g,1}^{ * \left( n \right)}}}{{1 + z_g^{*\left( n \right)}}}$;
\STATE Update $n = n + 1$ ;
\STATE \textbf{Until}: ${\rm{Judgemark}} = 1$

\STATE The optimal value of the problem is ${\lambda _{n - 1}}$;
\end{algorithmic}
\end{algorithm}


\begin{remark}
\textnormal{(Tightness of Approximation) It should be noted that the tightness of the Bernstein-type inequality highly relies on the selection of parameter $\sigma $. If we choose $\sigma = \ln \left( {\varepsilon ^{ - 1}} \right)$ in (\ref{e18}), the Bernstein-type inequality will act as a very conservative approximation, which deteriorates the secrecy performance of the system. Thus the value of $\sigma $ should be elaborately adjusted to achieve a high tightness of the approximation. The conservative property of the Bernstein-type inequality and the specific adjustment method will be discussed in detail in Section \uppercase\expandafter{\romannumeral3}-D.}
\end{remark}

\begin{remark}
\textnormal{(Feasibility) It should be pointed out that to satisfy the QoS requirements of all the QUs, there exists a feasible region shaped by these constraints, and the optimal solution only exists when the feasible region is not empty. Therefore, the target QoS thresholds and the scheduled QU's channel quality have a significant impact on the achievable secrecy performance. The QU scheduling scheme can be carefully designed to enlarge the feasible region of the optimization problem (\ref{P1}). Specifically, since the intra-cluster interference has been partially removed by SIC and the ZFBF vectors are designed according to the instantaneous CSI of the SUs, the inter-cluster interference will play a significant role on the SINRs of the QUs. Inspired by the above observations, an effective QU scheduling scheme is that the QUs in the same direction are scheduled as a cluster pairing with each SU, which aims to make the channel fading vectors of the SU and QUs in each cluster have higher correlation. Specifically, we schedule the QUs into the $g$-th cluster with the following condition:
\begin{equation}\label{en10}
\frac{{{\bf{h}}_{g,k}^H{{\bf{h}}_{g,1}}}}{{{\bf{h}}_{g,1}^H{{\bf{h}}_{g,1}}}} > {\phi _{g,k}},\;g = 1,2,...,G,\;k = 2,3,...,K,
\end{equation}
where ${{\phi _{g,k}}}$ represents the required threshold for the channel correlation coefficient between the instantaneous CSI of the $k$-th user and the SU in the $g$-th cluster (The subscript of ${{\phi _{g,k}}}$ is omitted in the following sections for notional simplicity).}

\end{remark}

\subsection{Proposed Method for the Upper Bound Case}
In the upper bound case, we substitute (\ref{en9}) into (\ref{P2e2}) to obtain the specific form of the SOP and then exploit the similar method to handle the SOP constraint. The transformed SOP constraint has the same form as (\ref{e15}) while
\begin{equation}\label{en16}
  {\bf{\Sigma }}_g^{\left( u \right)} = {{\bf{W}}_{g,1}} - {{\bf{W}}_{g,2}} - {{\bf{W}}_{g,3}},
\end{equation}
where ${{{\bf{W}}_{g,1}}} \!= \!{{\bf{w}}_g}{\bf{w}}_g^H\rho {\alpha _{g,1}}$, ${{{\bf{W}}_{g,2}}} \!=\! {z_g}{{\bf{w}}_g}{\bf{w}}_g^H\sum\nolimits_{k \ne 1}^K {\rho {\alpha _{g,k}}} $, and ${{{\bf{W}}_{g,3}}} \!=\! {z_g}\sum\nolimits_{i = 1,i \ne g}^G {{{\bf{w}}_i}{\bf{w}}_i^H\sum\nolimits_{k = 1}^K {\rho {\alpha _{i,k}}} } $. It can be easily recognized that ${\bf{\Sigma }}_g^{\left( u \right)}$ is a Hermite matrix. Then, we exploit the Bernstein-type inequality and obtain the same form as (\ref{e17}), where $\sigma $ should also be adjusted to reduce the conservatism as in Section \uppercase\expandafter{\romannumeral3}-D. Then, we can approximately transform the SOP constraint as the same form as (\ref{e18}). The max-min  problem can be formulated as
\begin{subequations}\label{Pnn1}
\begin{align}
&\mathop {\max }\limits_{{\bf{z}}} \mathop {\min }\limits_{1 \le g \le G} \;{\log _2}\left( {\frac{{1 + {{\left| {{\bf{h}}_{g,1}^H{{\bf{w}}_g}} \right|}^2}\rho {\alpha _{g,1}}}}{{1 + {z_g}}}} \right)\tag{25}\\
&\;\;\textrm{s.t.}\;\;(\textrm{\ref{P2e3}})~\textrm{and}~(\textrm{\ref{e18}}).\notag
\end{align}
\end{subequations}

However, since ${\bf{\Sigma }}_g^{\left( u \right)}$ consists of the product of variables ${\bf{z}}$ and ${\bf{a}}$ in ${{\bf{W}}_{g,2}}$ and ${{\bf{W}}_{g,3}}$, the right-hand side of (\ref{e18}) is not a linear function of optimization variables and consequently the transformed constraint is still non-convex. To deal with the non-convexity of the objective function and the transformed SOP constraint, we propose an AO-based algorithm and decouple the problem (\ref{Pnn1}) into two sub-problems.

\subsubsection{Sub-problem 1: Fix {\bf{a}} and Optimize {\bf{z}} }
In the $\left( {m + 1} \right)$-th iteration, when the optimal ${{\bf{a}}^{\left( m \right)}} = \left\{ {\alpha _{g,k}^{\left( m \right)},g = 1,2,...,G,k = 1,2,...,K} \right\}$ in the $m$-th iteration is obtained, we fix ${{\bf{a}}^{\left( m \right)}}$ and optimize ${\bf{z}} = \left\{ {{z_g},g = 1,2,...,G} \right\}$ to solve the following problem:
\begin{subequations}\label{Pn5}
\begin{align}
&\mathop {\max }\limits_{\bf{z}} \mathop {\min }\limits_{1 \le g \le G} \;{\log _2}\left( {\frac{{1 + {{\left| {{\bf{h}}_{g,1}^H{{\bf{w}}_g}} \right|}^2}\rho \alpha _{g,1}^{\left( m \right)}}}{{1 + {z_g}}}} \right), \ \
\textrm{s.t.}\;\;(\textrm{\ref{e18}})\tag{26}
\end{align}
\end{subequations}
Since problem (\ref{Pn5}) only has variable ${\bf{z}}$, the original product terms of ${\bf{z}}$ and ${\bf{a}}$ can be seen as the linear terms of variable ${\bf{z}}$ in this sub-problem and the constraints of problem (\ref{Pn5}) is thereby convex. By exploiting the monotonicity of logarithm, the objective function is transformed as
\begin{equation}\label{en20}
  \mathop {\min }\limits_{\bf{z}} \mathop {\max }\limits_{1 \le g \le G} \;\frac{{1 + {z_g}}}{{1 + {{\left| {{\bf{h}}_{g,1}^H{{\bf{w}}_g}} \right|}^2}\rho \alpha _{g,1}^{\left( m \right)}}}.
\end{equation}
After introducing the slack variable $\tau $, problem  (\ref{Pn5}) is transformed as follows:
\begin{subequations}\label{Pn6}
\begin{align}
&\mathop {\min }\limits_{{\bf{z}},\tau } \;\tau \label{Pn6e1} \\
&\;\;\textrm{s.t.}\;\;1 + {z_g} \le \tau \left( {1 + {{\left| {{\bf{h}}_{g,1}^H{{\bf{w}}_g}} \right|}^2}\rho \alpha _{g,1}^{\left( m \right)}} \right),\;\;(\textrm{\ref{e18}}), \label{Pn6e2}
\end{align}
\end{subequations}
which can be easily observed as a convex problem.

\subsubsection{Sub-problem 2: Fix {\bf{z}} and Optimize {\bf{a}} }
In the $\left( {m + 1} \right)$-th iteration, we fix the obtained ${{\bf{z}}^{\left( {m + 1} \right)}} = \left\{ {z_g^{\left( {m + 1} \right)},g = 1,2,...,G} \right\}$ and optimize ${\bf{a}} = \left\{ {{\alpha _{g,k}},g = 1,2,...,G,k = 1,2,...,K} \right\}$ by solving the following problem:
\begin{subequations}\label{Pn4}
\begin{align}
&\mathop {\max }\limits_{\bf{a}} \mathop {\min }\limits_{1 \le g \le G} \;{\log _2}\left( {\frac{{1 + {{\left| {{\bf{h}}_{g,1}^H{{\bf{w}}_g}} \right|}^2}\rho {\alpha _{g,1}}}}{{1 + z_g^{\left( {m + 1} \right)}}}} \right)\tag{29}\\
&\;\;\textrm{s.t.}\;\;(\textrm{\ref{P2e3}})~\textrm{and}~(\textrm{\ref{e18}}).\notag
\end{align}
\end{subequations}
Similarly, after introducing the slack variable $\tau $, we can transform the above problem to deal with the non-convexity of the objective function as follows:
\begin{subequations}\label{Pn3}
\begin{align}
&\mathop {\max }\limits_{\bf{a},\tau } \;\tau \label{Pn3e1}\\
&\;\;\textrm{s.t.}\;\;\;1 + {\left| {{\bf{h}}_{g,1}^H{{\bf{w}}_g}} \right|^2}\rho {\alpha _{g,1}} \ge {2^\tau }\left( {1 + z_g^{\left( {m + 1} \right)}} \right)\label{Pn3e2}\\
&\;\;\;\;\;\;\;\;\;(\textrm{\ref{P2e3}})~\textrm{and}~(\textrm{\ref{e18}}),\nonumber
\end{align}
\end{subequations}
which is a convex problem.

Since the objective value with the solutions obtained by solving the two optimization sub-problems is non-decreasing over iterations, and the optimal value of (\ref{Pn5}) is finite due to the power-limited NOMA system, the iterative solution is guaranteed to converge to a stationary solution. The proposed AO-based algorithm is summarized as \emph{Algorithm \ref{A3}}.

\begin{algorithm}[t]
\caption{AO-based Algorithm for the MMSR Problem in the Upper Bound Case}\label{A3}
\begin{algorithmic}[1]
\STATE \textbf{Initialize}: ${{\bf{a}}^{\left( 0 \right)}} = \left\{ {\alpha _{g,k}^{\left( 0 \right)},g = 1,2,...,G,k = 1,2,...,K} \right\}$, $m = 0$;
\STATE \textbf{Repeat}:
\STATE $m = m + 1$;
\STATE Solve problem (\ref{Pn6}) and obtain the optimal solution ${{\bf{z}}^{\left( m \right)}} = \left\{ {z_g^{\left( m \right)},g = 1,2,...,G} \right\}$ with ${{\bf{a}}^{\left( {m - 1} \right)}}$;
\STATE Solve problem (\ref{Pn3}) and obtain the optimal solution ${{\bf{a}}^{\left( m \right)}}$ with ${{\bf{z}}^{\left( m \right)}}$;
\STATE \textbf{Until}: Convergence
\end{algorithmic}
\end{algorithm}

\begin{remark}
\textnormal{(Initialization of ${{\bf{a}}^{\left( 0 \right)}}$) We herein propose an efficient approach to find an initial feasible solution for the
algorithm by reducing the feasible region and replacing the non-convex SOP constraint with the convex one. Specifically, the ${\bf{\Sigma }}_g^{\left( u
\right)}$ in the upper bound case can be replaced by the ${{\bf{\Sigma }}_g^{\left( l \right)}}$ in the lower bound case and we solve the similar problem like (\ref{P3}) by removing the denominator of the objective function. This is because in the practical scenario, if intra-cluster and inter-cluster interferences are removed at Eve (i.e., lower bound case), the receive SINR of the targeted signal at Eve is enhanced, which will lead to a worse secrecy performance. In other words, if the variables can satisfy the constraints in the lower bound case, these variables can also satisfy those in the upper bound case and we can obtain an initial ${{\bf{a}}^{\left( 0 \right)}}$ in the feasible region.}
\end{remark}

\subsection{Reducing the Conservatism of Approximation}
It is pointed out in Section \uppercase\expandafter{\romannumeral3}-B that the tightness of the Bernstein-type inequality highly relies on the selection of $\sigma $ and thus has a significant impact on the secrecy performance. Therefore, it is necessary to elaborately select the value of $\sigma $ to achieve high tightness of the approximation. Similarly to the setting of $\sigma$ in Section \uppercase\expandafter{\romannumeral3}-B, we select $\sigma = \ln \left( {\varepsilon _0^{ - 1}} \right)$ by introducing ${\varepsilon _0}$ to represent the maximum allowable SOP in the process of optimization. Then, after solving the approximated problems, the following lemma is introduced to calculate the actual SOP:
\begin{lemma}\label{SOP_Lemma}
  \textnormal{Let ${\lambda _i},i = 1,2,...,N$ denote the eigenvalues of ${\bf{Q}}$ in the descending order. If ${N_1}$ eigenvalues among $\left\{ {{\lambda _i}} \right\}$ are positive and distinct, then we have
  \begin{equation}\label{en21}
    \Pr \left\{ {{\bf{\tilde h}}_E^H{\bf{Q}}{{{\bf{\tilde h}}}_E} > z} \right\} = \sum\limits_{i = 1}^{{N_1}} {{{\prod\limits_{l \ne i}^N {\left( {1 - \frac{{{\lambda _l}}}{{{\lambda _i}}}} \right)} }^{ - 1}}} \exp \left( { - \frac{z}{{{\lambda _i}}}} \right).
  \end{equation}
  }
\end{lemma}
\begin{proof}
The detailed proof can be seen in \cite{ReferenceGaussian}.
\end{proof}

\begin{figure}[t]
\centering
\includegraphics[width=9cm,height=6cm]{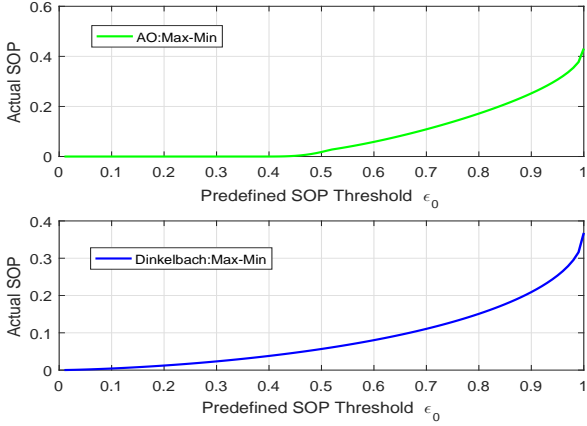}
\caption{Actual SOP versus predefined SOP parameter ${\varepsilon _0}$. Parameter settings: $G = N = 6$, $K = 3$, $a = 4$, ${\sigma _N^2} =  - 90dBm$, $\phi = 0.9$, $\varepsilon  = {10^{ - 2}}$, $r = 2$ for all QUs and $P = 40dBm$}\label{fbisection}
\end{figure}


According to (\ref{en21}), if the optimal values of the optimization variables have been obtained in the above subsections, we can exploit eigenvalues to calculate the actual SOP in the lower and upper bound cases by \emph{Lemma \ref{SOP_Lemma}}. To illustrate the conservatism of the approximation and obtain the relationship between the actual SOP and the parameter ${\varepsilon _0}$, we solve the optimization problems with ${\varepsilon _0}$ varying from 0.01 to 1 and depict the actually achieved SOP as the function of ${\varepsilon _0}$ in Fig. \ref{fbisection}. In the process of simulation to obtain Fig. \ref{fbisection}, according to the user scheduling requirement in (\ref{en10}), the specific instantaneous CSI of the $k$-th user in the $g$-th cluster can be modeled as
\begin{equation}\label{en8}
  {{\bf{h}}_{g,k}} = \sqrt {{\phi}} {{\bf{h}}_{g,1}} + \sqrt {1 - {\phi}} {{{\bf{\hat e}}}_{g,k}},\;k = 2,...,K,
\end{equation}
where ${\phi}$ represents the channel correlation coefficient as defined in Section \uppercase\expandafter{\romannumeral3}-C and ${{{\bf{\hat e}}}_{g,k}}$ reflects the difference of instantaneous CSI and is independent of ${{\bf{h}}_{g,1}}$ with independently and identical distribution (i.i.d) zero mean and unit variance complex Gaussian distributed entries. As can be seen in Fig. \ref{fbisection}, the actual value of SOP is far less than ${\varepsilon _0}$ in each optimization problem. Therefore, to achieve high tightness of the Berstein-type inequality, ${\varepsilon _0}$ is set to ensure the calculated actual value of SOP to be equal to $\varepsilon $. Since the actual value of SOP is a nondecreasing function of ${\varepsilon _0}$, we can use the bisection method to find out the specific value of ${\varepsilon _0}$.


\subsection{Computational Complexity Analysis}
\subsubsection{Lower Bound Case}In each iteration of \emph{Algorithm \ref{A1}}, problem (\ref{P3}) is a linear programming (LP) and the complexity of solving an LP is $O\left( {n_{LP}^2{m_{LP}}} \right)$, where ${{m_{LP}}}$ is the number of linear constraints and ${{n_{LP}}}$ is the dimension of optimization variables \cite{ReferenceMIMO1}. Specifically, we have ${m_{LP}} = G\left( {K + 1} \right) + 1$ and ${n_{LP}} = G\left( {K + 1} \right) + 1$. Therefore, the computational complexity of \emph{Algorithm \ref{A1}} is $O\left( {{\eta }{l_1}{{\left( {G\left( {K + 1} \right) + 1} \right)}^3}} \right)$, where ${{l_1}}$ denotes the number of iterations of the Dinkelbach algorithm and ${\eta}$ denotes the number of bisection iterations to achieve high tightness of the approximation as described in Section \uppercase\expandafter{\romannumeral3}-D.

\subsubsection{Upper Bound Case}The computational complexity of AO-based \emph{Algorithm \ref{A3}} includes the complexities of solving the subproblems (\ref{Pn6}) and (\ref{Pn3}). The above two subproblems both are LPs and the complexities are ${\varpi _1} = O\left( {2G{{\left( {G + 1} \right)}^2}} \right)$ and ${\varpi _2} = O\left( {{{\left( {GK + 1} \right)}^2}\left( {G\left( {K + 1} \right) + 1} \right)} \right)$ in each iteration, respectively. Then we can obtain that the computational complexity of \emph{Algorithm \ref{A3}} is $O\left( {{\eta}{l_2}\left( {{\varpi _1} + {\varpi _2}} \right)} \right)$, where ${{l_2}}$ denotes the number of alternating iterations.

\section{Max-Sum Achievable Secrecy Rate}
In this section, we design the optimal power allocation strategy in both lower and upper bound cases that maximizes the SSR of the SUs subject to the same constraints as in Section \uppercase\expandafter{\romannumeral3}. The important MSSR metric reflects the maximal ability to provide confidential transmissions.

\subsection{Problem Formulation}
Similar to the process in Section \uppercase\expandafter{\romannumeral3}-A, the MSSR problem is formulated as follows
\begin{subequations}\label{P4}
\begin{align}
&\mathop {\max }\limits_{{\bf{a}},{\bf{z}}} \sum\limits_{g = 1}^G {\left[ {{{\log }_2}\left( {\frac{{1 + {{\left| {{\bf{h}}_{g,1}^H{{\bf{w}}_g}} \right|}^2}\rho {\alpha _{g,1}}}}{{1 + {z_g}}}} \right)} \right]} \label{P4e1}\\
&\;\;\textrm{s.t.}\;\;{z_g}\! \ge\! Tr\!\left(\! {{\bf{\Sigma }}_g^{\left( i \right)}} \!\right)\!\! +\!\! \sqrt {2\sigma } {\left\|\! {{\bf{\Sigma }}_g^{\left( i \right)}} \!\right\|_F} \!\!+\!\! \sigma \!\! \cdot\!\! {\left\{\! {{\lambda _{\max }}\!\left(\! {{\bf{\Sigma }}_g^{\left( i \right)}} \!\right)} \!\right\}^ + },\label{P4e2}\\
&\;\;\;\;\;\;\;\;\;\frac{{{S_{g,k}}}}{{I_1^{g,k}\!\! +\!\! I_2^{g,k}\!\! +\!\! 1}} \!\ge \!{r_{g,k}},g = 1,...,G;\;k = 2,...,K\label{P4e3}\\
&\;\;\;\;\;\;\;\;\;\sum\limits_{g = 1}^G {\sum\limits_{k = 1}^K {{\alpha _{g,k}}} }  \le 1.\label{P4e4}
\end{align}
\end{subequations}
where $i \in \left\{ {l,u} \right\}$ in constraint (\ref{P4e2}) and the index $l,u$ correspond to the approximated SOP constraint in the lower and upper bound cases, respectively.

\subsection{MO-based Method for the Lower Bound Case}
Due to the non-convexity of the objective function, problem (\ref{P4}) is also challenging to be solved directly. To further transform the optimization problem, we first introduce the following definition of the general linear fractional programming (GLFP) \cite{ReferenceMonotonic}.
\begin{myDef}\label{GLFP_Definition}
  \textnormal{(\textbf{GLFP})  An optimization problem belongs to the class of GLFP if it can be represented by the following formulation:
  \begin{subequations}\label{GLFP_Problem}
  \begin{align}
  &{\rm{maximize}}\;\;\;\varphi \left( {\frac{{{f_1}\left( {\bf{x}} \right)}}{{{g_1}\left( {\bf{x}} \right)}},\frac{{{f_2}\left( {\bf{x}} \right)}}{{{g_2}\left( {\bf{x}} \right)}},...,\frac{{{f_m}\left( {\bf{x}} \right)}}{{{g_m}\left( {\bf{x}} \right)}}} \right)\\
  &\;\;\;\;\;\;\textrm{s.t.}\;\;\;\;\;\;\;\;{\bf{x}} \in {\bf{\Delta }},
  \end{align}
  \end{subequations}
  where the domain ${\bf{\Delta }}$ is a non-empty polytope\footnote{Polytope means the generalization to any dimension of polygon in two dimensions, polyhedron in three dimensions, and polychoron in four dimensions. The detailed explanation can be seen in \cite{ReferenceMonotonic}.} in ${{\bf{R}}^n}$ (the $n$-dim real domain), function $\varphi :{\bf{R}}_ + ^m \to {\bf{R}}$ is increasing on ${\bf{R}}_ + ^m$, and functions ${f_1},...,{f_m},{g_1},...,{g_m}:{\bf{\Delta }} \to {{\bf{R}}_{ +  + }}$ are positive-valued linear affine functions on ${\bf{\Delta }}$.}
\end{myDef}

It can be observed that ${1 + {{\left| {{\bf{h}}_{g,1}^H{{\bf{w}}_g}} \right|}^2}\rho {\alpha _{g,1}}}$ and ${1 + {z_g}}$ both are positive-valued linear affine functions and $\sum\nolimits_{g = 1}^G {{{\log }_2}\left(  \cdot  \right)} $ is increasing on ${\bf{R}}_ + ^m$. In addition, ${{\bf{\Sigma }}_g^{\left( l \right)}} = {{\bf{w}}_g}{\bf{w}}_g^H\rho {\alpha _{g,1}}$ in the lower bound case and thus all the constraints in (\ref{P4}) are linear functions of optimization variables, which means the feasible domain of (\ref{P4}) is a nonempty polytope if it exists. Thus, problem (\ref{P4}) in the lower bound case fits the characteristics of the GLFP form by \emph{Definition \ref{GLFP_Definition}} and the hidden monotonicity in the objective function can be exploited to recognize the specific GLFP as a monotonic optimization problem that can be solved by an efficient outer polyblock approximation algorithm \cite{Reference15} for the globally optimal solution (The basic mathematical preliminaries of monotonic optimization and the hidden monotonicity of GLFP will be shown in \textbf{Appendix B}). Specifically, to solve the problem (\ref{P4}) we first introduce the vector ${\bf{u}} = {[{u_1},{u_2},...,{u_G}]^T}$ with each element defined as
\begin{equation}\label{e20}
  {u_g} = \frac{{1 + {{\left| {{\bf{h}}_{g,1}^H{{\bf{w}}_g}} \right|}^2}\rho {\alpha _{g,1}}}}{{1 + {z_g}}}, g = 1,2,...,G.
\end{equation}
By exploiting the monotonic characteristic of the sum-log function, problem (\ref{P4}) can be transformed as the following standard monotonic optimization problem:
\begin{equation}\label{P6}
\mathop {\max }\limits_{\bf{u}} \;\sum\limits_{g = 1}^G {{{\log }_2}\left( {{u_g}} \right)},\ \ \textrm{s.t.}\;\;\;{\bf{u}} \in \Xi
\end{equation}
where $\Xi $ is defined as
\begin{equation}\label{e21}
  \Xi \! =\! \left\{\! {\left. {\bf{u}} \right|1\! \le\! {u_g}\! \le\! \frac{{1\! + \!{{\left| {{\bf{h}}_{g,1}^H{{\bf{w}}_g}} \right|}^2}\rho {\alpha _{g,1}}}}{{1 + {z_g}}}, 1 \!\le \!g \!\le\! G, \left\{ {{\bf{a}},{\bf{z}}} \right\}\! \in \!\Lambda } \!\right\},
\end{equation}
with $\Lambda $  being the feasible region determined by the constraints in problem (\ref{P4}). Since a monotonically increasing function always achieves its maximum over a polyblock at one of its proper vertices, the outer polyblock approximation algorithm successively maximizes the increasing objective function on a sequence of polyblocks that enclose the feasible set $\Xi$(The definition of the polyblock and the proper vertices of a polyblock will be shown in \textbf{Appendix B}). Based on this, we firstly construct an initialized polyblock ${B^{\left( 1 \right)}}$ with the vertex set ${V^{\left( 1 \right)}}$ that contains only one proper vertex ${{\bf{u}}^{\left( 1 \right)}}$. The constructed polyblock is supposed to cover the whole feasible set $\Xi$. Then we construct the new smaller polyblock ${B^{\left( 2 \right)}}$ with a new vertex set ${V^{\left( 2 \right)}}$ by replacing ${{\bf{u}}^{\left( 1 \right)}}$ with the new vertices $\left\{ {{\bf{\tilde u}}_1^{\left( 1 \right)},...,{\bf{\tilde u}}_G^{\left( 1 \right)}} \right\}$. Each new vertex is generated by
\begin{equation}
  {\bf{\tilde u}}_g^{\left( l \right)} = {{\bf{u}}^{\left( l \right)}} - \left( {u_g^{\left( l \right)} - {\varphi _g}\left( {{{\bf{u}}^{\left( l \right)}}} \right)} \right){{\bf{e}}_g},\label{e27}
\end{equation}
where ${\varphi _g}\left( {{{\bf{u}}^{\left( l \right)}}} \right)$ is the $g$-th element of $\Phi \left( {{{\bf{u}}^{\left( l \right)}}} \right)$, which is the projection of ${{\bf{u}}^{\left( l \right)}}$ on the feasible set $\Xi$, and ${{\bf{e}}_g}$ is the unit vector with the non-zero element only at index $g$. The projection of ${{\bf{u}}^{\left( l \right)}}$ on the feasible region is $\Phi \left( {{{\bf{u}}^{\left( l \right)}}} \right) = \lambda {{\bf{u}}^{\left( l \right)}}$ where
\begin{subequations}\label{e26}
\begin{align}
&\lambda  = \max \left\{ {\left. \beta  \right|\beta {\bf{u}} \in \Xi } \right\}\notag\\
&\;\; = \max \left\{\! {\left. \beta  \right|\beta \! \le \!\mathop {\min }\limits_{1 \le g \le G} \!\!\frac{{1 \!+\! {{\left| {{\bf{h}}_{g,1}^H{{\bf{w}}_g}} \right|}^2}\rho {\alpha _{g,1}}}}{{u_g^{\left( l \right)}\left( {1\! +\! {z_g}} \right)}},\left\{ {{\bf{a}},{\bf{z}}} \right\} \!\in\! \Lambda } \right\}\notag\\
&\;\;  = \mathop {\max }\limits_{\left\{ {{\bf{a}},{\bf{z}}} \right\} \in \Lambda } \mathop {\min }\limits_{1 \le g \le G} \;\frac{{1 + {{\left| {{\bf{h}}_{g,1}^H{{\bf{w}}_g}} \right|}^2}\rho {\alpha _{g,1}}}}{{u_g^{\left( l \right)}\left( {1 + {z_g}} \right)}}.\tag{39}
\end{align}
\end{subequations}
It is observed that problem (\ref{e26}) can also be solved by the Dinkelbach algorithm for the globally optimal solution. After calculating the projection and replacing the vertices, the feasible set $\Xi$ is still contained in the newly constructed smaller polyblock. Then we choose the optimal vertex whose projection maximizes the objective function of problem (\ref{P6}) by
\begin{equation}
  {{\bf{u}}^{\left( {l+1} \right)}} = \mathop {\arg \max }\limits_{{\bf{u}} \in {V^{\left( {l+1} \right)}}} \left\{ {\sum\limits_{g = 1}^G {{{\log }_2}\left( {{\varphi _g}\left( {\bf{u}} \right)} \right)} } \right\}.\label{e28}
\end{equation}
Then we repeat the above procedure to construct smaller polyblocks iteratively, and the algorithm terminates when a predefined tolerance is satisfied. The algorithm is summarized as \emph{Algorithm \ref{A2}}.

\begin{algorithm}[t]
\caption{Outer Polyblock Approximation Algorithm for the MSSR Problem in the Lower Bound Case}\label{A2}
\begin{algorithmic}[1]
\STATE Initialize the polyblock ${B^{\left( 1 \right)}}$ with the vertex set ${V^{\left( 1 \right)}} = \left\{ {{{\bf{u}}^{\left( 1 \right)}}} \right\}$ , where the elements of the ${{\bf{u}}^{\left( 1 \right)}}$ are
\begin{align}
&\;\;\;\;\;\;{u_g} = 1 + {{{\left| {{\bf{h}}_{g,1}^H{{\bf{w}}_g}} \right|}^2}\rho },\;\;\;\;g = 1,2...G.\nonumber
\end{align}
\STATE Set the allowable error tolerance $\delta  \ll 1$ and initialize the iteration index $l = 1$;
\STATE \textbf{Repeat}:
\STATE Construct the smaller polyblock ${B^{\left( {l + 1} \right)}}$ with ${V^{\left( {l + 1} \right)}}$ by replacing ${{\bf{u}}^{\left( l \right)}}$ with $\left\{ {{\bf{\tilde u}}_1^{\left( l \right)},...,{\bf{\tilde u}}_G^{\left( l \right)}} \right\}$. Each ${\bf{\tilde u}}_g^{\left( l \right)}$ is generated by (\ref{e27}), where $\Phi \left( {{{\bf{u}}^{\left( l \right)}}} \right) = \lambda {{\bf{u}}^{\left( l \right)}}$ and $\lambda$ is obtained by \textbf{Algorithm \ref{A1}} to solve the problem (\ref{e26});
\STATE Find ${{\bf{u}}^{\left( {l + 1} \right)}}$ by (\ref{e28})
\STATE Update $l = l + 1$;
\STATE \textbf{Until}:${{\left\| {{{\bf{u}}^{\left( l \right)}} - \Phi \left( {{{\bf{u}}^{\left( l \right)}}} \right)} \right\|} \mathord{\left/
 {\vphantom {{\left\| {{{\bf{u}}^{\left( l \right)}} - \Phi \left( {{{\bf{u}}^{\left( l \right)}}} \right)} \right\|} {\left\| {{{\bf{u}}^{\left( l \right)}}} \right\|}}} \right.
 \kern-\nulldelimiterspace} {\left\| {{{\bf{u}}^{\left( l \right)}}} \right\|}} \le \delta ;$
\STATE ${{\bf{u}}^*} = \Phi \left( {{{\bf{u}}^{\left( l \right)}}} \right)$ and the optimal solution of the problem is obtained by calculating $\Phi \left( {{{\bf{u}}^{\left( l \right)}}} \right)$
\end{algorithmic}
\end{algorithm}

\subsection{AO-based Method for the Upper Bound Case}
The max-sum optimization problem in the upper bound case can be formulated as
\begin{subequations}\label{P5}
\begin{align}
&\mathop {\max }\limits_{{\bf{a}},{\bf{z}}} \sum\limits_{g = 1}^G {\left[ {{{\log }_2}\left( {1 \!+ \!{{\left| {{\bf{h}}_{g,1}^H{{\bf{w}}_g}} \right|}^2}\rho {\alpha _{g,1}}} \right)\! -\! {{\log }_2}\left( {1\! +\! {z_g}} \right)} \right]}\tag{41} \\
&\;\;\textrm{s.t.}\;\;(\textrm{\ref{P4e2}}),~(\textrm{\ref{P4e3}}),~\textrm{and}~(\textrm{\ref{P4e4}}).\notag
\end{align}
\end{subequations}
To deal with the non-convexity of the objective function and the transformed SOP constraint in (\ref{P4e2}) with ${\bf{\Sigma }}_g^{\left( u \right)}$ defined in (\ref{en16}), we propose an AO-based algorithm and decouple problem (\ref{P5}) into the following two sub-problems.

\subsubsection{Sub-problem 1: Fix {\bf{a}} and Optimize {\bf{z}} }
In the $\left( {m + 1} \right)$-iteration, when the optimal ${{\bf{a}}^{\left( m \right)}} = \left\{ {\alpha _{g,k}^{\left( m \right)},g = 1,2,...,G,k = 1,2,...,K} \right\}$ in the $m$-th iteration is obtained, we fix ${{\bf{a}}^{\left( m \right)}}$ and optimize ${\bf{z}} = \left\{ {{z_g},g = 1,2,...,G} \right\}$ to solve the following problem
\begin{subequations}\label{Pn7}
\begin{align}
&\mathop {\min }\limits_{{\bf{z}}} \;\sum\limits_{g = 1}^G {{{\log }_2}\left( {1 + {z_g}} \right)},\ \ s.t. \ (\textrm{\ref{P4e2}}). \tag{42}
\end{align}
\end{subequations}
It can be observed that the $g$-th constraint in (\ref{P4e2}) is related to ${z_g}$ and ${\bf{\Sigma }}_g^{\left( u \right)}$, while ${\bf{\Sigma }}_g^{\left( u \right)}$ defined in (\ref{en16}) is only related to the power allocation policy and ${z_g}$. In addition, there is no constraint that requires the interaction of different ${z_g}$ with each other. Therefore, $G$ constraints in  (\ref{Pn7}) are independent with each other. Since the objective function is the sum of $G$ terms corresponding to ${z_g}$, the optimization of ${\bf{z}}$ can be decoupled into $G$ independent convex optimization problems after exploiting the monotonicity of the logarithmic function.

\subsubsection{Sub-problem 2: Fix {\bf{z}} and Optimize {\bf{a}} }
Based on the obtained value of optimization variables ${{\bf{z}}^{\left( {m + 1} \right)}} = \left\{ {z_g^{\left( {m + 1} \right)},g = 1,2,...,G} \right\}$ in the $\left( {m + 1} \right)$-th iteration, we fix ${{\bf{z}}^{\left( {m + 1} \right)}}$ and optimize ${\bf{a}} = \left\{ {{\alpha _{g,k}},g = 1,2,...,G,k = 1,2,...,K} \right\}$ to solve the following problem
\begin{subequations}\label{Pn10}
\begin{align}
&\mathop {\max }\limits_{\bf{a}} \;\sum\limits_{g = 1}^G {{{\log }_2}\left( {1 + {{\left| {{\bf{h}}_{g,1}^H{{\bf{w}}_g}} \right|}^2}\rho {\alpha _{g,1}}} \right)} \tag{43}\\
&\;\;\textrm{s.t.}\;\;(\textrm{\ref{P4e2}}),~(\textrm{\ref{P4e3}}),~\textrm{and}~(\textrm{\ref{P4e4}}).\notag
\end{align}
\end{subequations}
By introducing the slack variables ${\tau _g},g = 1,2,...,G$, the (\ref{Pn10}) can be transformed as
\begin{subequations}\label{Pn9}
\begin{align}
&\mathop {\max }\limits_{\bf{a},{\tau _g}} \;\sum\limits_{g = 1}^G \;\;{{\tau _g}}  \label{Pn9e1}\\
&\;\;\textrm{s.t.}\;\;\;1 + {\left| {{\bf{h}}_{g,1}^H{{\bf{w}}_g}} \right|^2}\rho {\alpha _{g,1}} \ge {2^{{\tau _g}}} \label{Pn9e2}\\
&\;\;\;\;\;\;\;\;\;(\textrm{\ref{P4e2}}),~(\textrm{\ref{P4e3}}),~\textrm{and}~(\textrm{\ref{P4e4}}),\nonumber
\end{align}
\end{subequations}
which can be easily observed as a convex problem.

Since the objective value obtained by solving the two optimization sub-problems is non-decreasing over iterations, and the optimal value of (\ref{Pn5}) is finite due to the power-limited NOMA system, the iterative solution is guaranteed to finally converges to a stationary solution. By exploiting the similar method to obtain the initial feasible solution in Section \uppercase\expandafter{\romannumeral3}-C, the proposed AO-based algorithm is summarized as \emph{Algorithm \ref{A4}}. It should be pointed out that we can take the same approach as in Section \uppercase\expandafter{\romannumeral3}-D to achieve high tightness of the approximated SOP constraint and solve the problem (\ref{P4}) in the lower and upper bound cases.

\begin{algorithm}[t]
\caption{AO-based Algorithm for the MSSR Problem in the Upper Bound Case}\label{A4}
\begin{algorithmic}[1]
\STATE \textbf{Initialize}: ${{\bf{a}}^{\left( 0 \right)}} = \left\{ {\alpha _{g,k}^{\left( 0 \right)},g = 1,2,...,G,k = 1,2,...,K} \right\}$, $m = 0$;
\STATE \textbf{Repeat}:
\STATE $m = m + 1$;
\STATE Solve problem (\ref{Pn7}) and obtain the optimal solution ${{\bf{z}}^{\left( m \right)}} = \left\{ {z_g^{\left( m \right)},g = 1,2,...,G} \right\}$ with ${{\bf{a}}^{\left( {m - 1} \right)}}$;
\STATE Solve problem (\ref{Pn9}) and obtain the optimal solution ${{\bf{a}}^{\left( m \right)}}$ with ${{\bf{z}}^{\left( m \right)}}$;
\STATE \textbf{Until}: Convergence
\end{algorithmic}
\end{algorithm}

\subsection{Computational Complexity Analysis}
\subsubsection{Lower Bound Case}Based on the complexity of solving a LP, the complexity of the projection is ${\varpi _3} = O\left( {{l_3}{{\left( {G\left( {K + 1} \right) + 1} \right)}^3}} \right)$, where ${{l_3}}$ denotes the number of iterations to calculate the projection. Then, the computational complexity of \emph{Algorithm \ref{A2}} is $O\left( {{\eta}{l_4}\left( {G + 1} \right){\varpi _3}} \right)$, where ${{\eta}}$ denotes the number of bisection iterations to achieve high tightness of the approximation, ${{l_4}}$ denotes the number of outer iterations in \emph{Algorithm \ref{A2}}, and ${\left( {G + 1} \right)}$ denotes the number of vertices that need to calculate the projection in each iteration.

\subsubsection{Upper Bound Case}Similar to Section \uppercase\expandafter{\romannumeral3}-E, the computational complexity of \emph{Algorithm \ref{A4}} is $O\left( {{\eta}{l_5}\left( {G + {{\left( {G\left( {K + 1} \right)} \right)}^2}\left( {G\left( {K + 1} \right) + 1} \right)} \right)} \right)$, where ${{l_5}}$ denotes the number of alternating optimization iterations.

\section{Simulation Results}
In this part, we numerically evaluate the average performance of our designed NOMA scheme through computer simulation. In our considered system, the numbers of the transmit antennas at the BS and user clusters are set as $N = G = 6$, and it is assumed that there are $K = 3$ legitimate users in each cluster without special instructions. The independent channel coefficient vector combined with small-scale fading and large-scale path loss from the BS to the $k$-th legitimate user in the $g$-th cluster is denoted by ${{\bf{h}}_{g,k}} = {{\bf{g}}_{g,k}}{d_{g,k}}^{ - \frac{a}{2}}$, where ${d_{g,k}}$ represents the distance, ${{\bf{g}}_{g,k}}$ denotes the Rayleigh fading channel vector, and $a = 4$ is the path-loss exponent{\footnote{We introduce the large-scale path loss in this section to make our simulation results more practical and reasonable, which has no effect on our aforementioned design of transmission scheme, user scheduling, and the optimization algorithm.}}. To satisfy the user scheduling requirement, the instantaneous ${{\bf{g}}_{g,k}}$ in each cluster is modeled in a similar way as (\ref{en8}). As for the distance between the legitimate users and the BS, we consider three ranges according to their minimum and maximum distances away from the BS and the parameters are set as $\left( {{{\rm{D}}_{{\rm{min}}}},{{\rm{D}}_{{\rm{max}}}}} \right) = \left\{ {\left( {50m,100m} \right),\left( {100m,200m} \right),\left( {200m,300m} \right)} \right\}$. For each cluster, the distances of the $K = 3$ scheduled users are uniformly distributed in these three ranges. Additionally, the distance between the BS and Eve is represented by ${{d_e} = 200m}$. Under this condition, the statistical CSI of Eve can be expressed as ${{\bf{h}}_e} \sim CN\left( {0,{\Gamma _e}{{\bf{I}}_N}} \right)$ where ${\Gamma _e} = d_e^{ - a}$. As for other constant parameters, the average receive noise power at LUs and Eve is set as ${\sigma _N^2} =  - 90 dBm$, the maximum tolerable calculation error and the predefined maximum iteration steps of the Dinkelbach algorithm and the outer polyblock approximation algorithm are both set as ${10^{ - 2}}$ and 50, respectively, and the optimization problem is solved for 100 times with randomly generated channel realizations.

\begin{figure}[t]
\centering
\includegraphics[width=9cm,height=6cm]{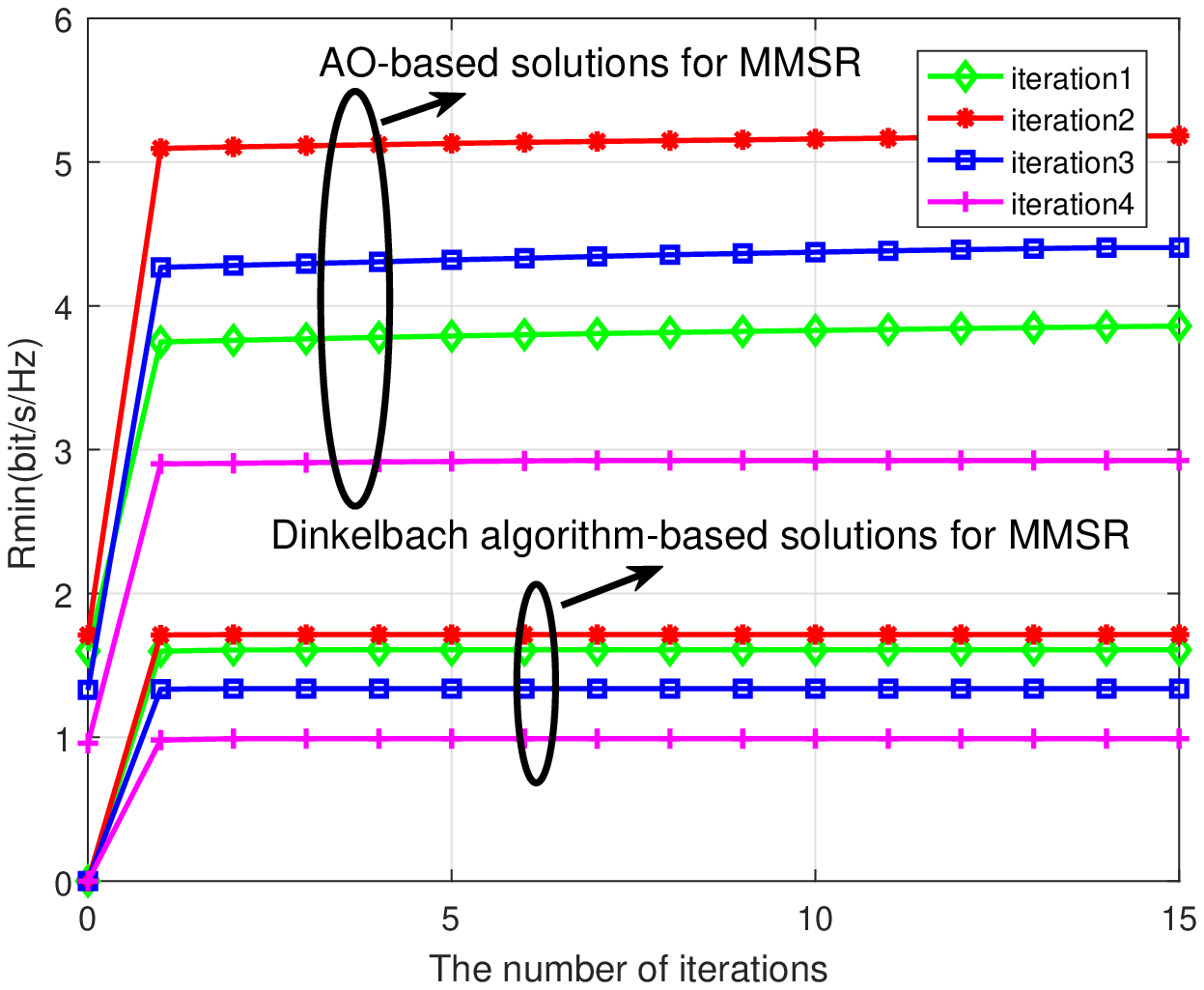}
\includegraphics[width=9cm,height=6cm]{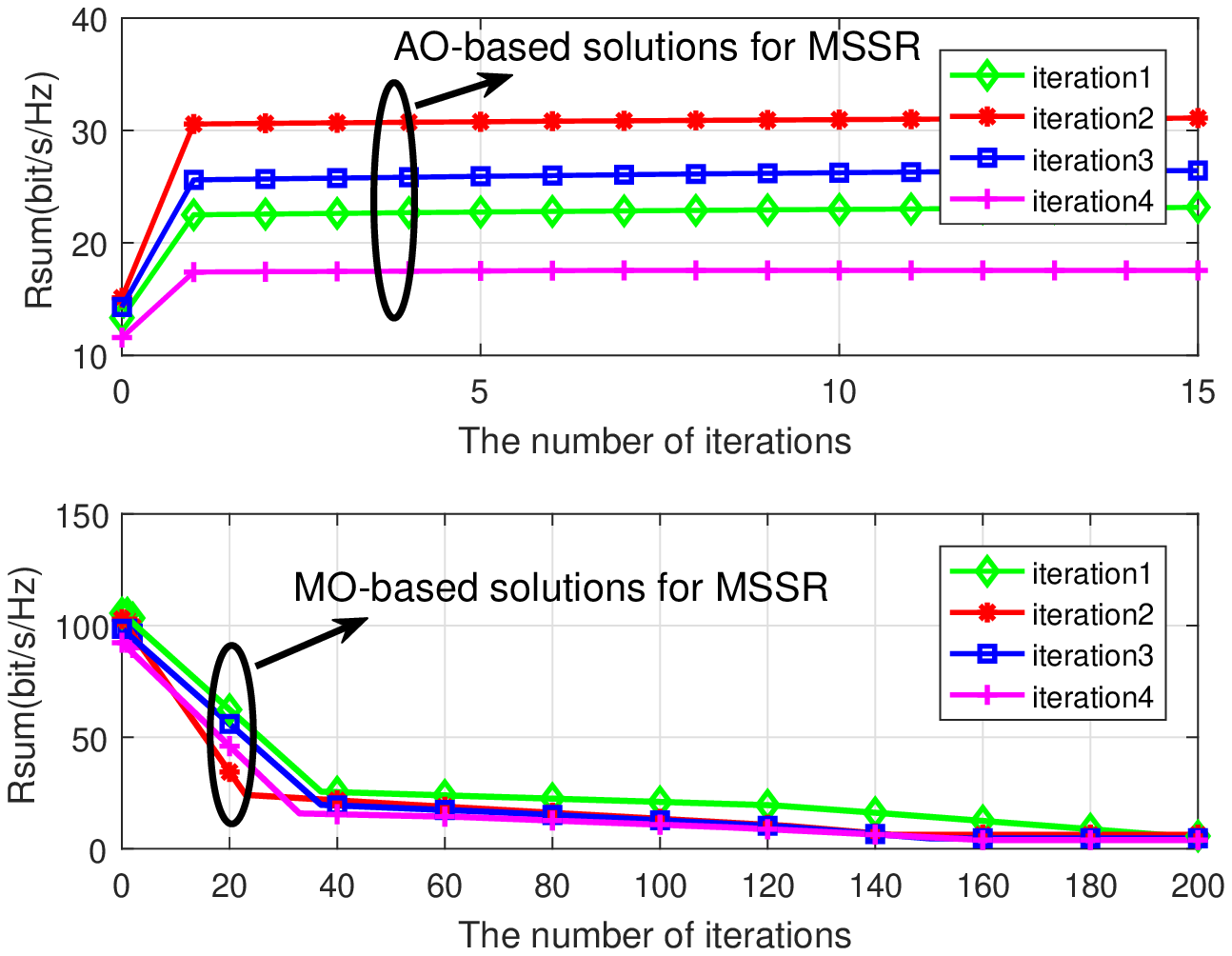}
\caption{Convergence of the proposed algorithms, where the channel correlation coefficient $\phi = 0.9$, $\varepsilon  = {10^{ - 2}}$, and $r = 2$ for all the QUs.}\label{f6}
\end{figure}

\subsection{Convergence of Proposed Algorithms}
Fig. \ref{f6} illustrates the convergence of the proposed algorithms to solve the MMSR and MSSR problems in the upper and lower bound cases. As can be seen from Fig. \ref{f6}, all of the proposed algorithms can converge to the optimal or suboptimal solutions with finite number of iterations. In addition, for each proposed algorithm, we notice that the required number of iterations to converge with randomly generated channel condition does not appear to be much different from each other, which implies that the speed of convergence of each proposed algorithm is not sensitive to the channel parameters. Moreover, we notice that the rate of convergence of the MO-based outer polyblock approximation algorithm is significantly slower than those of the other algorithms, which provides a trade-off between the secrecy performance and the rate of convergence.

It is also observed from Fig. 3 that only one iteration achieves convergence for AO-based and Dinkelbach algorithm-based solutions. For this observation, it should be pointed out that the considered downlink NOMA system is severely interference-limited. In addition, the SOP constraint of each SU and the QoS constraint of each QU are required to be satisfied. Therefore, the feasible regions of the proposed two optimization problems are relatively small. Under this condition, based on the obtained initial solutions, the difference of the objective function values in the two consecutive iterations is relatively small and thus only one iteration almost achieves convergence.

\begin{figure}[t]
\centering
\includegraphics[width=9cm,height=6cm]{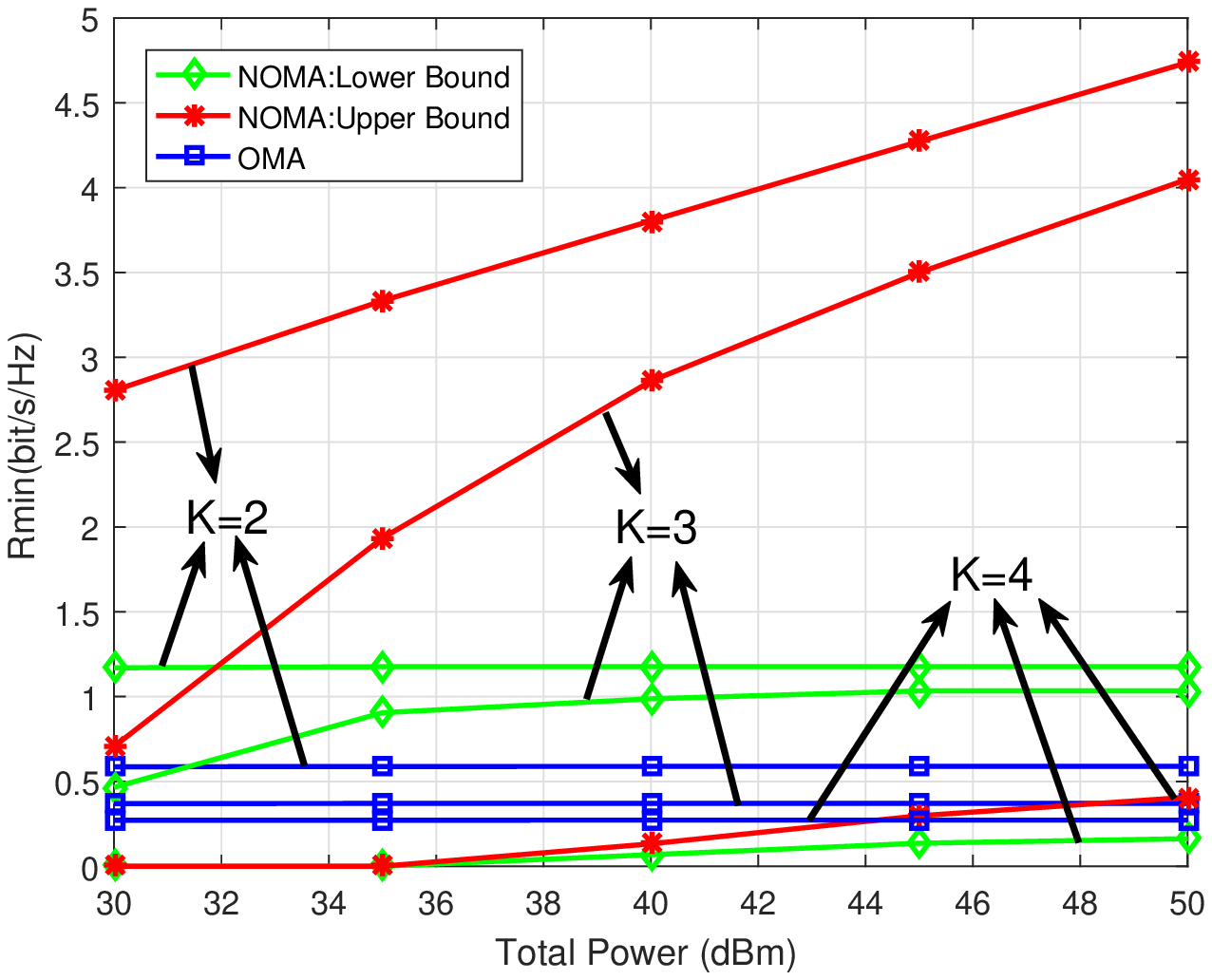}
\includegraphics[width=9cm,height=6cm]{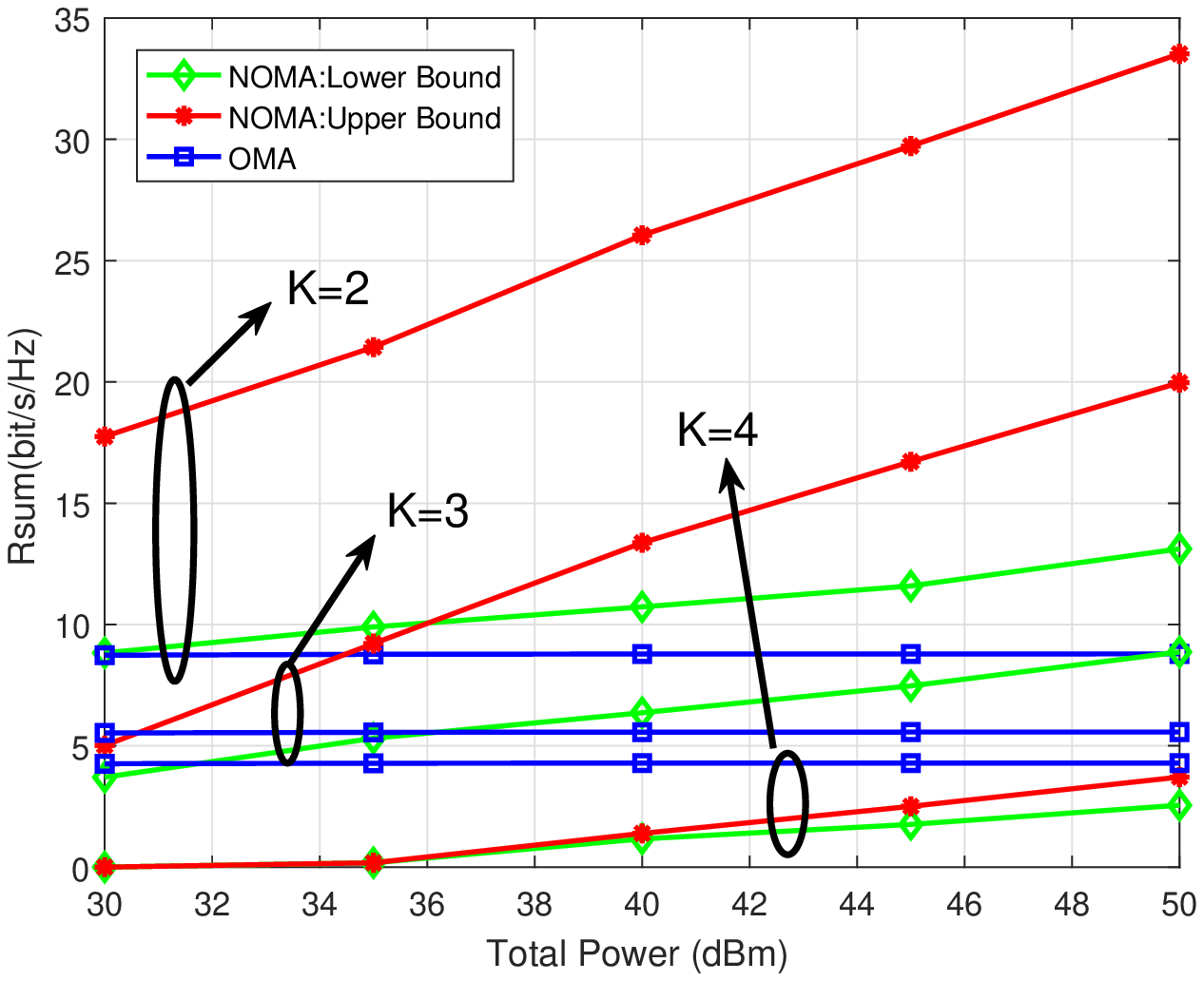}
\caption{Average MSR and SSR of SUs versus the total transmit power at the BS for different numbers of QUs in each cluster, where the channel correlation coefficient $\phi = 0.9$, $\varepsilon  = {10^{ - 2}}$, and $r = 2$ for all the QUs.}\label{f2}
\end{figure}

\subsection{Average Secrecy Rates Versus the Total Transmit Power and the Number of LUs}
In Fig. \ref{f2}, we investigate the average MSR and SSR of SUs versus the total transmit power at the BS for different numbers of QUs in each cluster\footnote{Under the condition that $K=2$, the distances of the scheduled users are uniformly distributed in $\left( {50m,100m} \right)$ and $\left( {100m,200m} \right)$. And under the condition that $K=4$, the distances of the scheduled users are uniformly distributed in $\left( {50m,100m} \right)$, $\left( {100m,200m} \right)$, $\left( {200m,300m} \right)$, and $\left( {300m,400m} \right)$.}, respectively. Under the condition that $K=2$ or $K=3$, it is observed that both our proposed two cases are superior to the conventional OMA one when the transmit power is not too low. Additionally, the secrecy performance of the upper bound case is improved obviously as the available transmit power increases. However, the secrecy performance of the lower bound case is just a little bit better and the secrecy performance of the OMA case is even almost unaffected with the increase of transmit power, which leads to the result that the performance gap between the upper bound case and the other two cases becomes larger. This observation can be reasonably accounted for due to the fact that the performance gap between the upper bound case and the other two cases reflects the impact of interference on the receive SINR of Eve. As the transmit power increases, more transmit power means more potential interference when Eve detects its targeted signal in the upper bound case, which leads to the increasing performance gap.

Furthermore, we notice that the secrecy performance of the conventional OMA case is superior to our proposed two cases when the transmit power is not sufficient high under the condition that $K=4$. This observation is due to the fact that the feasible region is decreased as the number of QUs in each cluster increases. If there is no feasible power allocation policy to satisfy all the QoS requirements, the achievable secrecy rate in this transmission is zero and the average secrecy performance is highly degraded under this condition.


\begin{figure}[t]
\centering
\includegraphics[width=9cm,height=6cm]{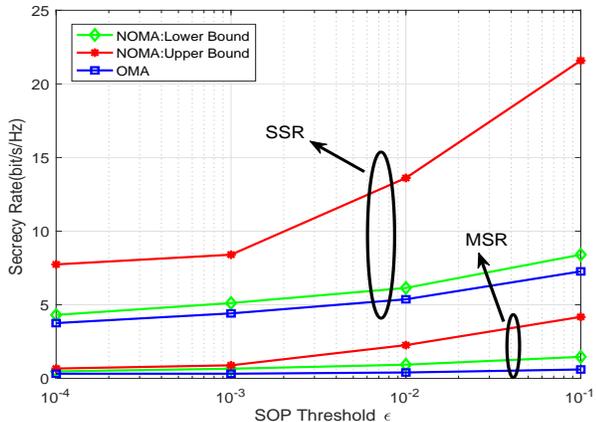}
\caption{Average MSR and SSR of SUs versus the maximum allowable SOP threshold, where the channel correlation coefficient $\phi = 0.9$, $P = 40dBm$, and $r = 2$ for all the QUs.}\label{f3}
\end{figure}


\subsection{Average Secrecy Rates Versus the Required SOP Threshold}
Fig. \ref{f3} depicts the average MSR and SSR of SUs versus the maximum allowable SOP threshold $\varepsilon $, respectively. It is observed that our proposed two cases are superior to the conventional OMA one for the both secrecy rate metrics. Additionally, since the allowable SOP threshold affects the redundancy term ${\log _2}\left( {1 + {z_g}} \right)$ in the objective function as described in Section \uppercase\expandafter{\romannumeral3}-A, the secrecy performance of each case becomes better with more relaxed SOP requirements. Moreover, it is revealed according to the simulation results that the performance gap between the three cases becomes larger with more relaxed SOP requirements.


\begin{figure}[t]
\centering
\includegraphics[width=9cm,height=6cm]{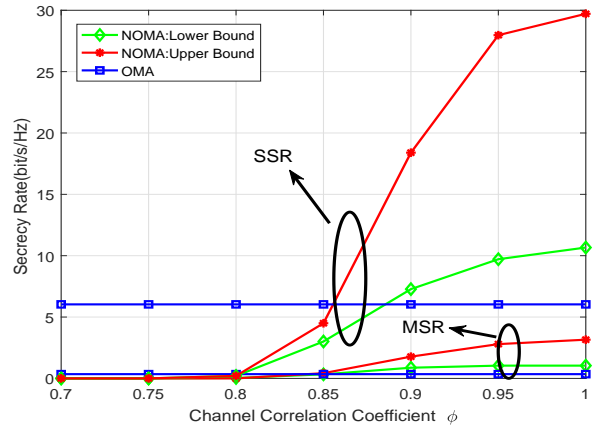}
\caption{Average MSR and SSR of SUs versus the channel correlation coefficient, where $\varepsilon  = {10^{ - 2}}$, $P = 40dBm$, and $r = 2$ for all the QUs.}\label{f4}
\end{figure}


\subsection{Average Secrecy Rates Versus the Channel Correlation Coefficient}
Fig. \ref{f4} depicts the average MSR and SSR of the SUs versus the channel correlation coefficient $\phi$ in (\ref{en10}) and (\ref{en8}), respectively. It is observed that the secrecy performance of the conventional OMA case is unaffected, which can be easily explained by the fact that there is no inter-cluster interference due to the orthogonal characteristics. It can also be observed that if $\phi$ is relatively low, the secrecy performances of our proposed two NOMA cases are worse than that of the conventional OMA scheme and becomes better with the increase of $\phi$, due to the fact that the inter-cluster interference for the QUs can be more sufficiently eliminated. Under this condition, the QoS requirements of the QUs can be satisfied with higher probability and the extra transmit power can be exploited to improve the secrecy performance. According to this observation, it is pointed out that the efficient user scheduling is important in our proposed NOMA scheme, since the achievable secrecy rate falls to zero if the channel correlation coefficient is relatively small.

\begin{figure}[t]
\centering
\includegraphics[width=9cm,height=6cm]{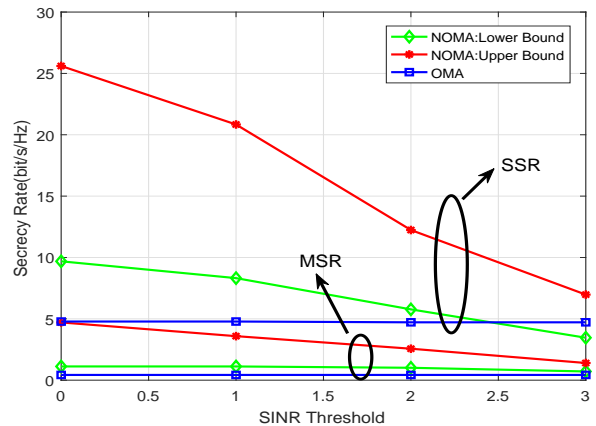}
\caption{Average MSR and SSR of SUs versus the required SINR threshold, where $\varepsilon  = {10^{ - 2}}$, $P = 40dBm$, and the channel correlation coefficient $\phi = 0.9$. }\label{f5}
\end{figure}


\subsection{Average Secrecy Rates Versus the Required SINR Threshold}
Fig. \ref{f5} depicts the average MSR and SSR of the SUs versus the required SINR threshold, respectively. It is observed that the secrecy performance of the conventional OMA case is unaffected, which can be easily explained by the fact that the secrecy performance of SUs is independent with the QoS requirements of the QUs due to the orthogonal characteristics of the OMA scheme. It can also be observed that the secrecy performances of our proposed two NOMA cases are better than that of the conventional OMA case but gradually becomes worse as the SINR threshold increases. This observation can be explained in two aspects. On one hand, the QoS requirements of the QUs are more difficult to satisfy with the increase of the SINR threshold and thereby there may not be extra transmit power for the SUs to achieve secrecy performance. On the other hand, if the QoS requirements of the QUs are satisfied, the increase of the SINR threshold leads to more power allocated to the QUs and less power allocated to the SUs. Therefore, the secrecy rate performance is degraded. According to this observation, low-rate service users, such as IoT receivers, are more suitable to act as the QUs in our proposed NOMA cases.

\section{Conclusion}
In this paper, we have comprehensively investigated and optimized both the lower and upper bounds of the secrecy performance in a downlink MISO NOMA system with legitimate users categorized as security-required and QoS-required ones. A dynamic user scheduling and grouping strategy has been proposed and efficient optimization algorithms have been adopted to achieve the globally optimal and sub-optimal solutions for the MMSR and MSSR problems in the lower and upper bound cases. It has been shown that the proposed NOMA cases can have sufficient superiority over the conventional OMA one in the practical scenario, and some key parameters, such as the QoS requirements and the required threshold for channel correlation coefficients, can significantly affect the secrecy performance of the proposed NOMA cases.

\begin{appendices}

\section{Proof of Lemma 1}
Since ${\bf{Q}} \in {{\rm{C}}^{N \times N}}$ is a Hermite matrix, the eigenvalue decomposition of ${\bf{Q}}$ can be expressed as ${\bf{Q}} = {{\bf{U}}^H}{\bf{\Lambda U}}$, where ${\bf{U}}$ is the unitary matrix and ${\bf{\Lambda}}$ is the diagonal matrix consisting of the eigenvalues of the matrix ${\bf{Q}}$. It is noted that since ${\bf{Q}}$ is a Hermite matrix, the eigenvalues of ${\bf{Q}}$ are all real numbers. Then, ${\bf{G}} = {{\bf{h}}^H}{\bf{Qh}}$ can be reformulated as
\begin{equation}\label{EquAppendix1}
    {\bf{G}} = {{\bf{h}}^H}{{\bf{U}}^H}{\bf{\Lambda Uh}} = {\left( {{\bf{Uh}}} \right)^H}{\bf{\Lambda }}\left( {{\bf{Uh}}} \right).
\end{equation}
Since ${\bf{U}}$ is the unitary matrix, ${\bf{\hat h}} = {\bf{Uh}}$ has the same distribution as ${\bf{h}}$ which means ${\bf{\hat h}}\sim CN\left( {0,{{\bf{I}}_N}} \right)$. Then based on (\ref{EquAppendix1}), ${\bf{G}}$ can be further transformed as
\begin{equation}\label{EquAppendix2}
    {\bf{G}} = \sum\limits_{n = 1}^N {{q_n}\left( {\hat h_{n,R}^2 + \hat h_{n,I}^2} \right)},
\end{equation}
where ${{q_n}}$ denotes the $n$-th real eigenvalue of the matrix ${\bf{Q}}$ and ${{\hat h}_{n,R}}$ and ${{\hat h}_{n,I}}$ represent the real and imaginary parts of the $n$-th element of ${\bf{\hat h}}$, respectively. According to the properties of circular symmetric complex gaussian distribution, ${{\hat h}_{n,R}}$ and ${{\hat h}_{n,I}}$ are independent with each other and both of them satisfy the distribution of $N\left( {0,1/2} \right)$. Therefore, we can reformulate ${\bf{G}}$ as
\begin{equation}\label{EquAppendix3}
    {\bf{G}} = \sum\limits_{n = 1}^{2N} {{{\hat q}_n}r_n^2}  = \sum\limits_{n = 1}^{2N} {\frac{1}{2}{{\hat q}_n}{{\left( {\sqrt 2 {r_n}} \right)}^2}},
\end{equation}
where ${r_n}\sim N\left( {0,1/2} \right), n = 1,2,...,2N$ are independent real random gaussian variables and
\begin{align}
&{{\hat q}_n}=\left\{
             \begin{array}{lcl}
             {q_n},\;\;\;\;\;\;\;\;\;n = 1,2,...,N \\
             {q_{n - N}},\;\;\;\;n = N + 1,N + 2,...,2N.
             \end{array}
        \right.\label{EquAppendix4}
\end{align}
Finally by applying Lemma (0.1) in \cite{Reference16} and exploiting the relations that $\sum\nolimits_{n = 1}^N {{q_n}}  = Tr({\bf{Q}})$ and $\sum\nolimits_{n = 1}^N {q_n^2}  = \left\| {\bf{Q}} \right\|_F^2$, \emph{Lemma 1} is proved.


\section{Mathematical Preliminaries}

In this appendix, we will provide some basic preliminaries of monotonic optimization and show the hidden monotonicity of GLFP \cite{ReferenceMonotonic}.
\begin{myDef}\label{Definition1}
  \textnormal{(\textbf{Increasing functions}) A function $f:R_ + ^n \to R$ is increasing if $f\left( {\bf{x}} \right) \le f\left( {\bf{y}} \right)$ when $0 \le {\bf{x}} \le {\bf{y}}$.}
\end{myDef}
\begin{myDef}\label{Definition2}
  \textnormal{(\textbf{Boxes}) If ${\bf{a}} \le {\bf{b}}$, then box $\left[ {{\bf{a}},{\bf{b}}} \right]$ is the set of all ${\bf{x}} \in {R^n}$ satisfying ${\bf{a}} \le {\bf{x}} \le {\bf{b}}$.}
\end{myDef}
\begin{myDef}\label{Definition3}
  \textnormal{(\textbf{Normal sets}) A set $\Omega  \subset R_ + ^n$ is normal if for any point ${\bf{x}} \in \Omega $, all other points ${{\bf{y}}}$ such that ${\bf{0}} \le {{\bf{y}}} \le {\bf{x}}$ are also in set $\Omega $. In other words, $\Omega  \subset R_ + ^n$ is normal if ${\bf{x}} \in \Omega  \Rightarrow \left[ {{\bf{0}},{\bf{x}}} \right] \subset \Omega $.}
\end{myDef}
\begin{myDef}\label{Definition4}
  \textnormal{(\textbf{Conormal sets}) A set $\Psi $ is conormal if ${\bf{x}} \in \Psi $ and ${{\bf{y}}} \ge {\bf{x}}$ implies ${{\bf{y}}} \in \Psi $. The set is conormal in $\left[ {{\bf{0}},{\bf{b}}} \right]$ if ${\bf{x}} \in \Psi  \Rightarrow \left[ {{\bf{x}},{\bf{b}}} \right] \subset \Psi $.}
\end{myDef}
\begin{myDef}\label{Definition5}
  \textnormal{(\textbf{Normal hull}) The normal hull of a set $\Phi  \subset R_ + ^n$ is the smallest normal set containing $\Phi $. Mathematically, the normal hull is given by $N\left( \Phi  \right) = { \cup _{{\bf{Z}} \in \Phi }}\left[ {{\bf{0}},{\bf{Z}}} \right]$.}
\end{myDef}
\begin{myDef}\label{Definition6}
  \textnormal{(\textbf{Canonical monotonic optimization formulation}) Monotonic Optimization is concerned with problems of the form $\max \left\{ {\left. {f\left( {\bf{x}} \right)} \right|{\bf{x}} \in \Omega  \cap \Psi } \right\}$, where $f\left( {\bf{x}} \right):R_ + ^n \to R$ is an increasing function, $\Omega  \subset \left[ {{\bf{0}},{\bf{b}}} \right] \subset R_ + ^n$ is a compact normal set with nonempty interior, and $\Psi $ is a closed conormal set on$\left[ {{\bf{0}},{\bf{b}}} \right]$.}
\end{myDef}
\begin{myDef}\label{Definition7}
  \textnormal{(\textbf{Upper boundary}) A point ${{\bf{\bar x}}}$ of a normal closed set $\Omega $ is called an upper boundary point of $\Omega $ if $\Omega  \cap \left\{ {\left. {{\bf{x}} \in R_ + ^n} \right|{\bf{x}} > {\bf{\bar x}}} \right\} = \emptyset $. The set of all upper boundary points of $\Omega $ is called its upper boundary and denoted by ${\partial ^ + }\Omega $.}
\end{myDef}
\begin{myDef}\label{Definition8}
  \textnormal{(\textbf{Polyblocks}) A set ${\rm{P}} \subset R_ + ^n$ is called a polyblock if it is a union of a finite number of boxes $\left[ {{\bf{0}},{\bf{Z}}} \right]$, where ${\bf{Z}} \in \Gamma $ and $\left| \Gamma  \right| <  + \infty $. The set $\Gamma $ is the vertex set of the polyblock.}
\end{myDef}
\begin{myDef}\label{Definition9}
  \textnormal{(\textbf{Proper vertices of a polyblock}) Let $\Gamma $ be the vertex set of a polyblock ${\rm{P}} \subset R_ + ^n$. A vertex ${\bf{v}} \in \Gamma $ is said to be proper if there is no ${\bf{\hat v}} \in \Gamma $ such that ${\bf{\hat v}} \ne {\bf{v}}$ and ${\bf{\hat v}} \ge {\bf{v}}$. A vertex is said to be improper if it is not proper. Improper vertices can be removed from the vertex set without affecting the shape of the polyblock.}
\end{myDef}

In general, the constraints derived from practical systems may result in an arbitrarily shaped feasible set. Then the following proposition shows that this kind of problem can still be formulated into the canonical form as long as $f\left( {\bf{x}} \right)$ is increasing.
\begin{pp}\label{PP1}
\textnormal{If $\Phi $ is an arbitrary nonempty compact set on $R_ + ^n$ and $\Omega  = N\left( \Phi  \right)$ is the normal hull of $\Phi $, then the problem $\max \left\{ {\left. {f\left( {\bf{x}} \right)} \right|{\bf{x}} \in \Phi } \right\}$ is equivalent to $\max \left\{ {\left. {f\left( {\bf{x}} \right)} \right|{\bf{x}} \in \Omega } \right\}$.}
\end{pp}

Based on the above definitions, the predefined problem (\ref{GLFP_Problem}) in Section \uppercase\expandafter{\romannumeral4}-B is equivalent to $\max \left\{ {\left. {\phi \left( {\bf{y}} \right)} \right|{\bf{y}} \in {\bf{u}}\left( {\bf{\Delta }} \right)} \right\}$, which can be further written as $\max \left\{ {\left. {\phi \left( {\bf{y}} \right)} \right|{\bf{y}} \in \Omega } \right\}$ by \emph{Proposition \ref{PP1}}, where $\Omega  = N\left( {{\bf{u}}\left( {\bf{\Delta }} \right)} \right) = \left\{ {\left. {{\bf{y}} \in R_ + ^m} \right|{\bf{y}} \le {\bf{u}}\left( {\bf{x}} \right),{\bf{x}} \in {\bf{\Delta }}} \right\}$. Since ${{\bf{u}}\left( {\bf{x}} \right)}$ is continuous on ${\bf{\Delta }}$, ${\bf{u}}\left( {\bf{\Delta }} \right)$ is compact. Thus, its normal hull $\Omega $ is also compact and is contained in box $\left[ {{\bf{0}},{\bf{b}}} \right]$. Furthermore, since all ${u_i}\left( {\bf{x}} \right)$ are positive, $\Omega $ has a nonempty interior. By this, we conclude that the transformed problem is a monotonic optimization problem in the canonical form.

\end{appendices}

\end{document}